\documentclass[10pt,reqno]{amsart}
\usepackage{amsmath,amssymb,enumerate}

\newcommand{\dt}{\dfrac{\rm d~}{{\rm d}t}}

\newcommand{\R}{\mathbb R}
\newcommand{\E}{\mathbb E}
\newcommand{\N}{\mathbb N}
\newcommand{\e}{\mathrm{e}}                
\newcommand{\de}{\mathop{}\!\mathrm{d}}  
\newcommand{\pa}{\mathop{}\!\partial} 
\newcommand{\Tr}{\mathop{}\!\mathrm{Tr}\mathop{}} 
\newcommand{\dive}{\mathop{}\!\mathrm{div}\,} 
\newcommand{\Id}{\mathop{}\!{\bf I}} 
\newcommand{\eps}{\varepsilon}
\newcommand{\ve}[1]{{\mathbf {#1}}}
\newcommand{\vev}[1]{{\boldsymbol {#1}}}
\newcommand{\fcr}{\mathcal X}
\newcommand{\Cal}[1]{{\mathcal {#1}}}

\newtheorem{theorem}{Theorem}[section]
\newtheorem{corollary}{Corollary}[section]
\newtheorem{lemma}{Lemma}[section]
\newtheorem{proposition}{Proposition}[section]
\theoremstyle{definition}
\newtheorem{definition}{Definition}[section]

\theoremstyle{remark}
\newtheorem{remark}{Remark}[section]
\numberwithin{equation}{section}

\begin{document}

\title[Inertial spin model for flocks and related kinetic equations]{Some aspects of the inertial spin model for flocks and related kinetic equations}

\author[D. Benedetto]{Dario Benedetto}
\address{Dario Benedetto \hfill\break \indent
	Dipartimento di Matematica, 
	Sapienza Universit\`a di Roma,
	\hfill\break \indent
	P.le Aldo Moro 5, 00185 Roma, Italy}
\email{benedetto@mat.uniroma1.it}
\author[P.\ Butt\`a]{Paolo Butt\`a}
\address{Paolo Butt\`a\hfill\break \indent
	Dipartimento di Matematica, 
	Sapienza Universit\`a di Roma,
	\hfill\break \indent
	P.le Aldo Moro 5, 00185 Roma, Italy}
\email{butta@mat.uniroma1.it}
\author[E. Caglioti]{Emanuele Caglioti}
\address{Emanuele Caglioti \hfill\break \indent
	Dipartimento di Matematica, 
	Sapienza Universit\`a di Roma,
	\hfill\break \indent
	P.le Aldo Moro 5, 00185 Roma, Italy}
\email{caglioti@mat.uniroma1.it}

\begin{abstract}
In this paper we study the macroscopic behavior of the inertial spin (IS) model. This model has been recently proposed to describe the collective dynamics of flocks of birds, and its main feature is the presence of an auxiliary dynamical variable, a sort of internal spin, which conveys the interaction among the birds with the effect of better describing the turning of flocks. After discussing the geometrical and mechanical properties of the IS model, we show that, in the case of constant interaction among the birds, its mean-field limit is described by a non-linear Fokker-Planck equation, whose equilibria are fully characterized. Finally, in the case of non-constant interactions, we derive the kinetic equation for the mean-field limit of the model in the absence of thermal noise, and explore its macroscopic behavior by analyzing the mono-kinetic solutions.
\end{abstract}



\maketitle
\thispagestyle{empty}

\section{Introduction}
\label{sec:intro}

The coordinated motions of group of animals like flocks of birds or schools of fishes are paradigmatic examples of the emergence of collective behavior. Further examples  include other group of animals, such as quadruped herds or insect swarm, and also human systems like crowds or communication networks.

Starting with the pioneering paper \cite{vicsek} (Vicsek model), several models have been proposed to explain the evolution of these systems \cite{CDFSTB,CKFL,CS1,CS2,deg-3d,LRC,LLK,PE}. In the simplest models \cite{CS1,CS2,vicsek}, a bird is modelled as a self-propelling particle that interacts with its neighbours. The interaction is such that neighbouring birds tends to align their velocities.

Kinetic equations, see, e.g., \cite{BFOZ,CDP,HT}, and hydrodynamic equations \cite{EdVL,PE,TB,TBL} have been proposed to describe these systems, with the aim to better capture their collective motion.

Recently, Cavagna et al.~\cite{ism} proposed a model of flocking which is different from the models quoted above because of the presence of an additional variable, called ``spin'',  which is related to the curvature of the trajectory of the bird. Roughly speaking, the effect of this spin variable is that the other birds do not influence directly the time derivative of the velocity of the bird but indirectly through its second derivative. According to the authors, this allows to better describe the turning of a flock and the propagation of the information along the flock in agreement with the experimental data \cite{top1,top2}.

The main aim of this paper is to derive a macroscopic description of the model proposed in \cite{ism}. In particular, several properties of the model discussed in \cite{ism} are here presented in a more precise mathematical setting.

In Sec.~2, we present the model, by focusing mainly on its geometrical and mechanical structures.  In Sec.~3, we consider the model in the absence of thermal noise and assuming that the interaction among the birds is independent of their spatial positions. We describe and slightly improve some alignment results recently obtained in \cite{ha2019}. In Sec.~4, we consider the case  when the noise is present and the interaction among the birds is constant. We show that, in the mean-field limit, the collective behavior of this system is well approximated by a suitable non-linear Fokker-Planck equation. In Sec.~5, we characterize the stationary (equilibrium) states for this equation. In particular, we show that below a critical temperature the system exhibits a phase transition for which the average velocity of the birds is non-zero.

In Sec.~6, we consider the mean-field limit of the model without noise in the case of non-constant interactions, and we derive a kinetic equation for it. Starting from this description we derive, formally, macroscopic equations of motion, which are associated to the mono-kinetic solutions of the kinetic equation in a suitable zero-range limit of the interaction. This analysis is also carried out in the case of topological interactions. An important feature of these mono-kinetic equations is that the linear perturbation of stationary states satisfies a wave-like equation. This fact is in qualitative agreement with the considerations made in \cite{ism}. Finally, we find and discuss the two particular classes of solutions given by the plane stationary rotating solutions and the stationary flows along a fixed curve in the space.

\section{The Inertial Spin Model}
\label{sec:ism}

The particle system proposed in \cite{ism} describes the motion of $N$ particles (birds), whose position and velocity at time $t$ are hereafter denoted by $\ve x_i(t) \in \R^3$ and $\ve v_i(t) \in \R^3$, $i=1,\ldots,N$, respectively. The main feature of the model is that the motion of the $i$-th particle is driven by an internal variable, the ``spin'' $\ve s_i(t)\in \R^3$, whose variation in time is in turn determined by a weighted mean velocity of the particles in a neighborhood of $\ve x_i(t)$ and by the interaction with a thermal bath. More precisely, the so-called inertial spin (IS) model consists in the following system of It\^o stochastic differential equations,
\[
\left\{
\begin{aligned}
&\de \ve x_i = \ve v_i\de t\,, \\
&\de \ve v_i = \frac 1{\chi} \ve s_i \wedge \ve v_i\de t\,,\\
&\de \ve s_i = \ve v_i \wedge  \left(
\frac {J}{v^2} \ve w_i\de t - \frac \eta {v^2}
\de \ve v_i + \frac 1v \de \vev \xi_i\right)\,,
\end{aligned}
\right.
\]
where
\begin{itemize}
\item $v>0$ is the scale of the velocities;
\item $\chi>0$ is an inertial coefficient;
\item $J>0$ is the strength of the interaction;
\item $\ve w_i =  \sum_{j=1}^N n_{ij} \ve v_j$ is a weighted mean of the velocities of the particles around the $i$-th particle, given by the matrix of  ``communication weights'' $\{n_{ij}\}$, with positive entries which may depend only on the positions $\{\ve x_k\}_{k=1}^N$;
\item the thermal bath interactions are given by the independent noises  $\vev \xi_i\in \R^3$, $\delta$-correlated in time, and a frictional term of coefficient $\eta/v^2>0$.
\end{itemize}
As we will see in Sec.~\ref{sec:vs.bagno}, from It\^o formula it is easy to show that $|\ve v_i|$ and $\ve v_i \cdot \ve s_i$ are conserved quantities. Indeed, the model is designed to evolve initial data with $|\ve v_i|=v$ and $\ve v_i \cdot \ve s_i = 0$ for all $i$, see \cite{ism}. Therefore, in what follows we always assume that $|\ve v_i(0)|=v$ and $\ve v_i(0) \cdot \ve s_i(0) = 0$ for all $i=1,\ldots,N$.\footnote{We could avoid the explicit dependence on the initial datum in the equations of motion by replacing $v$ with  $|\ve v_i|$, the solution being the same, but we prefer to keep the actual, simpler form. Concerning the other assumption, the analysis of Sec.~\ref{sec:vs-atttrito} actually covers also the case in which $\ve v_i \cdot \ve s_i \ne 0$.}

Since the inertial coefficient $\chi$ and the variables $\ve s_i$ cannot be observed separately, it is convenient to scale $\ve s_i \to \chi\ve s_i$ and, consequently, $(J,\eta) \to (\chi J, \chi\eta)$. Moreover, to emphasize the mean-field character of the interaction, we replace $n_{ij}$ by $n_{ij}/N$.

Thus, the equations of motion expressed in the rescaled variables take the following form,
\begin{equation}
\label{eq:particelle}
\left\{
\begin{aligned}
& \de \ve x_i  = \ve v_i \de t\,,\\
& \de \ve v_i= \ve s_i \wedge \ve v_i\de t\,,\\
& \de\ve s_i = \ve v_i \wedge \left( \frac J{v^2} \ve w_i \de t - \frac \eta{v^2} \de \ve v_i  + \frac 1v \sqrt{2\nu} \de \ve B_i\right), \\
& \ve w_i = \frac 1N \sum_{j=1}^N n_{ij} \ve v_j \,,\\
& |\ve v_i| = v\,, \quad \ve v_i \cdot \ve s_i = 0\,,
\end{aligned}
\right.
\end{equation}
where $\nu$ is a diffusive coefficient and $\ve B_i$, $i=1,\ldots, N$, are $N$ independent standard Brownian motions in $\R^3$.

In the rest of this section, we neglect the stochastic term and discuss the geometrical and mechanical aspects of the deterministic model. To this end, we notice that setting $\nu=0$ Eqs.~\eqref{eq:particelle} reduce to the following system of ordinary differential equations,
\begin{equation}
\label{eq:particelled}
\left\{
\begin{aligned}
& \dot {\ve x}_i  = \ve v_i\,,\\
& \dot {\ve v}_i  = \ve s_i \wedge \ve v_i\,,\\
& \dot{\ve s}_i  = \frac 1{v^2} \ve v_i \wedge (J \ve w_i - \eta\, \ve s_i \wedge \ve v_i)\,, \\
& \ve w_i  = \frac 1N \sum_{j=1}^N n_{ij} \ve v_j\,,\\
& |\ve v_i| = v\,, \quad \ve v_i \cdot \ve s_i = 0\,.
\end{aligned}
\right.
\end{equation}
Under suitable regularity and boundedness assumptions on $n_{ij}$, the solutions to Eqs.~\eqref{eq:particelled} exists globally in time and are uniquely determined by the initial data.

Suppose $\{\ve x_i(t),\ve v_i(t),\ve s_i(t)\}_{i=1}^N$ is a solution to Eqs.~\eqref{eq:particelled}. Since $|\ve v_i(t)|=v$, the scaled variable $tv$ is the arc-length of the curve $t\mapsto \ve x_i(t)$, so that
\[
\dot{\ve v}_i(t) =  v^2 \kappa_i(t)\ve n_i(t)\,,
\]
with $\ve n_i(t)$ the normal vector and $\kappa_i(t)$ the curvature of the curve $t\mapsto \ve x_i(t)$. Comparing this equation with Eq.~\eqref{eq:particelled}$_2$, and using the orthogonality relation $\ve v_i(t)\cdot \ve s_i(t)=0$, we derive the geometrical meaning of the spin variables $\ve s_i(t)$, i.e.,
\[
\ve s_i(t) = v \kappa_i(t) \ve b_i(t)\,,
\]
where $\ve b_i(t)$ is the bi-normal to the curve $t\mapsto \ve x_i(t)$.

\subsection*{The underlying Lagrangian Dynamics} 
\label{sec:2.1} 

In this section, we consider the conservative case when also the friction is neglected, i.e., Eqs.~\eqref{eq:particelled} with $\eta=0$. In \cite{ism}, the authors derive the IS model as a pseudo-Hamiltonian system in the variables $\{\ve v_i, \ve s_i\}_{i=1}^N$, starting from the 2-dimensional case. Here, we deepen the analysis on the structure of the model by showing that it is indeed a Lagrangian system.

The phase space of the $i$-th particle is $\R^3\times TS_v^2$, where $S_v^2$ denotes the spherical surface of radius $v$ and centre the origin, and $TS_v^2  = \cup_{\ve v\in S_v^2} T_{\ve v} S_v^2$ is its tangent bundle, i.e., the collection of pairs $(\ve v, \ve s)$ with $\ve v\in S_v^2$ and $\ve s \in T_{\ve v}S_v^2$, where $T_{\ve v}S_v^2$ denotes the tangent space to $S_v^2$ in $\ve v$ . We note that from the relations $\dot {\ve v}_i = \ve s_i \wedge  \ve v_i$ and $\ve s_i \cdot \ve v_i =0$ it follows that 
\[
\ve s_i = \frac {\ve v_i \wedge \dot {\ve v}_i}{v^2}\,,\qquad |\ve s_i| = \frac{|\dot {\ve v}_i|}v\,.
\]
In particular, $v^2\ve s_i$ can be seen as the angular momentum with respect to the origin of a ``particle'' of position $\ve v_i\in S_v^2$ and velocity $\dot {\ve v}_i$. By insisting on this interpretation, we calculate the acceleration of the particle,
\[
\ddot {\ve v}_i = \dot {\ve s}_i \wedge \ve v_i + \ve s_i \wedge \dot{\ve v}_i  = - \frac J{v^2} \ve v_i \wedge (\ve v_i \wedge \ve w_i) + \ve s_i \wedge (\ve s_i \wedge \ve v_i)\,. 
\]
Since $|\ve s_i| = |\dot {\ve v}_i|/v$ and $\ve v_i \cdot \ve s_i=0$, the second term is equal to $-v^{-2} |\dot {\ve v}_i|^2 \ve v_i$. Therefore, calling
\[
\ve F_i = J \ve w_i
\]
the ``force'' acting on the $i$-th particle due to the interaction with the others, we obtain
\begin{equation}
\label{eq:ddv}
\ddot {\ve v}_i  = P_i^\perp \ve F_i  - \frac{|\dot {\ve v}_i|^2}{v^2}  \ve v_i\,,
\end{equation}
where, setting $\hat {\ve v}_i = \ve v_i / |\ve v_i|$,
\begin{equation}
\label{proj}
P^\perp_i = \Id - \hat {\ve v}_i \otimes \hat {\ve v}_i
\end{equation}
is the projection operator on the tangent space $T_{\ve v_i}S_v^2$, which is orthogonal to $\ve v_i$.

\begin{theorem}
\label{thm:1}
Assume the communication matrix $\{n_{ij}\}$ is symmetric and let $\{\ve x_i(t),\ve v_i(t),\ve s_i(t)\}_{i=1}^N$ be a solution to Eqs.~\eqref{eq:particelled}. Then $\{\ve v_i(t)\}_{i=1}^N$ is a motion of the constrained mechanical system with (time-dependent) Lagrangian
\[
\mathcal{L}(\ve v_1,\dots \ve v_N, \dot{\ve v}_1,\dots \dot{\ve v}_N,t) = \frac 12 \sum_{i=1}^N \dot {\ve v}_i^2 - U(\ve v_1,\dots \ve v_N,t)
\]
and holonomic constraints $|\ve v_i|=v$, $i=1,\ldots,N$, where the ``potential energy'' $U$ is given by\footnote{Note that in general $n_{ij}$ depend on time throughout the space positions $\{\ve x_i(t)\}_{i=1}^N$.}
\[
U(\ve v_1,\dots \ve v_N,t) =  \frac J{4N} \sum_{i,j=1}^N n_{ij}(\ve v_i - \ve v_j)^2
\]

\end{theorem}

\begin{proof}
Since the communication matrix is symmetric we have $\ve F_i = -\nabla_{\ve v_i} U$. On the other hand, the term $-v^2|\dot {\ve v}_i|^2\ve v_i$ is exactly the (ideal) constraint reaction. The theorem is thus proved.
\end{proof}

According to Theorem \ref{thm:1}, the variables $\{\ve v_i\}_{i=1}^N$ can be viewed as the positions of $N$ particles constrained on the sphere $S_v^2$, which interact among each other by means of elastic forces whose intensities depend on the configuration $\{\ve x_i\}_{i=1}^N$. The energy of this system is
\begin{equation}
\label{E}
E= \frac 12 \sum_{i=1}^N |\dot {\ve v}_i|^2 + U
\end{equation}
(which is a conserved quantity when the entries $n_{ij}$ do not depend on the positions $\{x_k\}_{k=1}^N$, and hence $U$ is time independent).

The Lagrangian nature of the motion of the variables $\{\ve v_i\}_{i=1}^N$ should be compared with the original derivation of the model in \cite{ism}, where $\{\ve v_i, \ve s_i\}_{i=1}^N$, although not canonical coordinates, give rise to a pseudo-Hamiltonian structure, in which the derivatives of a Hamiltonian function are combined with the cross product. More precisely, letting
\[
H = \frac 12 \sum_i |\ve s_i|^2 + U\,,
\]
the evolution of these coordinates can be reshaped in the following form,
\[
\left\{
\begin{aligned}
\dot {\ve v}_i & = \nabla_{\ve s_i} H \wedge \ve v_i\,,\\
\dot {\ve s}_i & = \nabla_{\ve v_i} H \wedge \ve v_i\,.
\end{aligned}
\right.
\]
It is worthwhile to notice that $\nabla_{\ve v_i} \cdot (\nabla_{\ve s_i} H \wedge \ve v_i) = \nabla_{\ve s_i} \cdot (\nabla_{\ve v_i} H \wedge \ve v_i) = 0$; this fact can be useful in the kinetic description of the system.

\begin{remark}
\label{rem:1}
The system can also be written in Hamiltonian form using local coordinates on $S_v^2$. Moreover,
\[
H = \frac 12 \sum_i \ve p_i \cdot P_i^\perp \ve p_i + U
\]
is a Hamiltonian function, which gives the same second order equations for the variables $\ve v_i$, without using local coordinates. But in this case, the kinetic term is only positive semi-definite, and $\ve p_i \cdot \ve v_i$ is increasing in time. For our purposes, it is better to work with velocities and spins.
\end{remark}

\begin{remark}
\label{rem:2}
We finally spend few words on the case of general initial data with $\ve v_i(0) \cdot \ve s_i(0) = \alpha_i \ne 0$. We then have $\ve v_i(t) \cdot \ve s_i(t)=\alpha_i$ for any time $t$. After introducing the new variables
\[
\vev \sigma_i = \ve s_i - \frac{\alpha_i \ve v_i}{v^2}\,,
\]
we observe that
\[
\dot {\ve v}_i = \vev \sigma_i \wedge \ve v_i\,, \qquad \ve v_i \cdot \vev \sigma_i = 0\,, \qquad \vev \sigma_i =  \frac {\ve v_i \wedge \dot {\ve v}_i}{v^2}\,,
\]
so that the equations of motion become
\[
\left\{
\begin{aligned}
\dot {\vev \sigma}_i & = \frac{\ve v_i \wedge \ve F_i + \alpha_i \ve v_i \wedge  \vev \sigma_i}{v^2}\,, \\ \ddot {\ve v}_i  & =  P_i^\perp \ve F_i   - \frac{|\dot {\ve v}_i|^2}{v^2} \ve v_i  + \alpha_i  \frac {\ve v_i \wedge \dot {\ve v}_i}{v^2}\,.
\end{aligned}
\right.
\]
Therefore, the initial conditions $\ve v_i \cdot \ve s_i=\alpha_i$ modify the motion of the variables $\ve v_i$ by adding a sort of Lorentz force, generated by a (only locally defined) vector potential. 

Also the motion of the spatial coordinates is different with respect to the case $\alpha_i=0$. For instance, consider the case without inter-particle interaction ($J=0$). If $\alpha_i=0$ then $\ve v_i$ describes a free motion on $S_v^2$, i.e., it moves with constant speed along some great circle on $S_v^2$, and therefore also $\ve x_i$ moves around a circle. Otherwise, $\ve v_i$ moves along a circle around $\hat {\ve s}_i$, with $\ve v_i \cdot \hat {\ve s}_i = \alpha_i /|\ve s_i|$, so that in this case $\ve x_i$ moves along a cylindrical helix.
\end{remark}

\subsection*{The interaction}
\label{sec:2.2}

The matrix of the communication weights $n_{ij}$ drives the interaction among the particles. In literature, several choices of such weights have been proposed. In this paper, we will consider examples from two different classes.

\begin{itemize}
\item[(D)] The weights depend on the distance, i.e., 
\begin{equation}
\label{eq:D}
n_{ij} = \frac{K_{ij}(|\ve x_i -\ve x_j|)}{n_i^q}\,,
\end{equation}
where $K_{ij}=K_{ji}\colon \R^+ \to \R^+$ are positive non increasing functions with compact support, $n_i = \frac 1N \sum_{j\neq i} K_{ij}(|\ve x_i - \ve x_j|)$ is a normalization factor of the order of the local density, $q \in [0,1]$ is an exponent which modulates the dependence of the interaction on the local density: if $q= 1$ the intensity of the interaction does not depend on the local density, and $\ve w_i$ is a weighted mean velocity of the particles around  $\ve x_i$; if $q =0$ the interaction grows with the local density (this is the choice made in \cite{ism}).
\item[(R)] The weights depend on the rank (topological interaction), i.e.,
\begin{equation}
\label{eq:R}
n_{ij} = T(M_{|\ve x_i|,|\ve x_i - \ve x_j|})\,,
\end{equation}
where $T \colon [0,1] \to \R^+$ are positive non increasing functions and
\[
M_{|\ve x_i|,R} = \frac 1N \sum_k \fcr\{|\ve x_k - \ve x_i| < R\}
\]
is (proportional to) the mass of the particles contained in the sphere of radius $R$ around $\ve x_i$. It is worthwhile to notice that the integer $NM_{|\ve x_i|,|\ve x_i - \ve x_j|}$ is equal to the position of the particle $j$ in the ranking of the closer particles to $\ve x_i$.
\end{itemize}

\begin{remark}
\label{rem:3}
We notice that in the conservative case, Eqs. \eqref{eq:particelled} with $\eta=0$, the total spin $\sum_{i=1}^N \ve s_i$ is a conserved quantity if the weights are of class (D) with $q=0$, i.e., if $n_{ij}$ depends only on the distance $|\ve x_i - \ve x_j|$.
\end{remark}

\section{The free space deterministic case}
\label{sec:vs-atttrito}

In order to explore the properties of the IS model it is useful to consider the simpler case in which $n_{ij}$ are constant, so that the evolution of velocities and spins is independent of the positions of the particles. For $n_{ij}$ chosen in the class $(D)$ and $(R)$, this means that $K_{ij}$ or $T$ are constant functions, respectively. But even in the general case, as long as all the particles stay sufficiently close to each other, $K_{ij}$ or $T$ can be considered approximately constant. Although such condition could be satisfied only for finite time, the analysis of the approximated system can anyway highlight some general aspects of this kind of dynamics. 

In particular, in \cite{ha2019}, the authors consider Eqs.~\eqref{eq:particelled} with multiplicative constant communication rates, $n_{ij} = n_i n_j$, so that the average velocity is $\ve w_i = n_i \ve w$, with $\ve w = \frac 1N \sum_j n_j \ve v_j$ independent of $i$. In this case, disregarding the evolution of the position of particles, the equations of motion for the velocities and spins are
\begin{equation}
\label{eq:vs-attrito}
\left\{
\begin{aligned}
&\dot {\ve v}_i = \ve s_i \wedge \ve v_i\,,\\
&\dot {\ve s}_i = \frac 1{v^2}\ve v_i \wedge (J n_i \ve w - \eta \ve s_i \wedge \ve v_i)\,,\\
&\ve w =  \frac 1N \sum_{j=1}^N n_j \ve v_j \,,\\
& |\ve v_i| = v\,, \quad \ve v_i \cdot \ve s_i = \alpha_i\,,
\end{aligned}
\right.
\end{equation}
where, with respect to Eqs.~\eqref{eq:particelled}, we allow initial data with arbitrary values $\alpha_i\in\R$ of the conserved quantities $\ve v_i \cdot \ve s_i$.  In \cite{ha2019}, it is proved that, under suitable conditions on the initial datum, the evolution \eqref{eq:vs-attrito} with $\alpha_i=0$ for each $i=1,\ldots,N$, exhibits flocking in the sense that, as $t\to+\infty$, $|\ve s_i(t)| = v^{-1}|\dot {\ve v}_i(t)|\to 0$ and $|\ve v_i(t)-\ve v_j(t)|\to 0$. In this section, we extend the result, showing that the individual velocities $\ve v_i(t)$ converge separately, and extending the analysis to the general case in which $\alpha_i\ne 0$.\footnote{A warning to the reader: in \cite{ha2019} the intensity of the velocities are normalized fixing $|\ve v_i|=1$, while the inertial coefficient $\chi$ is not adsorbed as done here in deducing Eqs.~\eqref{eq:particelle}. Therefore,  one has to set $v=1$ here and $\chi=1$ in \cite{ha2019} when comparing our results with those of \cite{ha2019}.}

\smallskip
Let us frame the results in \cite{ha2019} and ours in the context of the theory of mechanical systems with friction. It is easy to show that the Eqs.~\eqref{eq:vs-attrito} have infinitely many stationary solutions, of different kind.

Given $\ve v\in S_v^2$ and a partition $I^+\cup I^- = \{1,\dots N\}$ of the set of indexes, an \textit{aligned stationary solution} is given by
\begin{equation}
\label{eq:sol-aligned}
\ve v_i = \pm \ve v \quad \forall\, i\in I^\pm\,, \qquad  \ve s_i = \frac{\alpha_i}{v^2} \ve v_i \quad \forall\,i=1,\ldots,N\,.
\end{equation}
If $I^-$ is empty, all the velocities are equal and the corresponding solutions are called \textit{flocking stationary solutions}. They correspond to minima of the potential energy $U$ (in which $U$ vanishes). 

The other kind of stationary solutions are the values of $\{\ve v_i\}_{i=1}^N$ such that $\ve w=0$, and $\ve s_i = \alpha_iv^{-2} \ve v_i$, which we call \textit{incoherent stationary solutions}. Apart for the 3D-2D difference, this decomposition of the set of the stationary solutions is the same as that for the stationary solutions of the Kuramoto model, see, e.g., \cite{BCM}. This is not a case, since the Kuramoto model can be obtained as the ``zero inertia'' limit of the planar motions of the IS model \cite{ha2019}. Note that if $\eta=0$ there are incoherent quasi-periodic solutions with $\ve w=0$ for all times. In these cases, the system is partitioned into planar subsystems. In each subsystem, all the particles have the same spin, which is orthogonal to the plane where the particles lie, and the mean velocity is zero. Therefore, each subsystem performs a uniform rotation on its plane. 

In the case $\eta>0$, the energy is dissipated by the friction and we expect that any solution tends to some equilibrium solution as $t\to +\infty$, but this result does not follow from general theorems, such as the Barbashin-Krasovskii-LaSalle principle, since the set of equilibria is the union of manifolds in the phase space in which $U$ is constant.

\begin{theorem}
\label{thm:2}
Assume $n_i>0$ for any $i=1,\dots ,N$, and let $\{\ve v_i(t),\ve s_i(t)\}_{i=1}^N$ be a solution to Eqs.~\eqref{eq:vs-attrito}. Then, setting $\ve w(t) =  \frac 1N \sum_{j=1}^N n_j \ve v_j(t)$, there exists the limit
\begin{equation}
\label{winfty}
\lim_{t\to+\infty}|\ve w(t)| = w_\infty
\end{equation}
and 
\begin{itemize}
\item if $w_\infty  = 0$ then $\{\ve v_i(t)\}_{i=0}^N$ converges to an incoherent stationary solution;
\item if $w_\infty>0$ then there exists the limit $\displaystyle \lim_{t\to+\infty}\frac {\ve w(t)}{|\ve w(t)|} = \ve u_{\infty}\in S_1^2$ and, for each $i=1,\dots,N$,
\[
\lim_{t\to+\infty}\ve v_i(t) = v\ve u_{\infty}\quad \text{or}\quad  \lim_{t\to+\infty}\ve v_i(t) = -v\ve u_{\infty}\,,
\]
i.e., $\{\ve v_i(t)\}_{i=0}^N$ converges to an aligned stationary solution.
\end{itemize}
\end{theorem}

\begin{remark}
\label{rem:4}
In \cite{ha2019}, only the case $\alpha_i=0$ is considered, and it is shown solely that whenever $|\ve w(t)|\to 0$ the solution converges to a stationary solution.
\end{remark}

To prove Theorem \ref{thm:2} we need two preliminary lemmas. First of all, we rewrite Eqs.~\eqref{eq:vs-attrito} by means of the variables $\vev \sigma_i = \ve s_i - (\alpha_i/v^2) \ve v_i = (\ve v_i \wedge \dot {\ve v}_i)/v^2$ introduced in Remark \ref{rem:2}. Since $\ve v_i \wedge (\ve s_i \wedge \ve v_i) = v^2 \vev \sigma_i$ we have,
\begin{equation}
\label{eq:sigma-sistema}
\left\{
\begin{aligned}
&\dot {\ve v}_i = \vev \sigma_i \wedge \ve v_i\,,\\
&\dot {\vev \sigma}_i = \frac{Jn_i \ve v_i \wedge \ve w
-\alpha_i \vev \sigma_i \wedge \ve v_i}{v^2} -\eta {\vev \sigma}_i \,,\\
&\ve w =  \frac 1N \sum_{j=1}^N n_j\ve v_j\,,\\
\end{aligned}
\right.
\end{equation}
where, in this case, $Jv^{-2}n_i\ve w = -\nabla_{\ve v_i} U$ with potential energy 
\begin{equation}
\label{eq:sigma-potenziale}
U(\ve v_1,\dots \ve v_N) = \frac J{4Nv^2} \sum_{ij} n_i n_j (\ve v_i - \ve v_j)^2 = \frac{J}{2N}\left(\sum_i n_i\right)^2 - \frac {JN}{2v^2} |\ve w|^2\,.
\end{equation}
The corresponding total energy
\begin{equation}
\label{eq:sigma-energia}
H(\ve v_1,\dots \ve v_N,\vev \sigma_1,\dots \vev \sigma_N) = \frac 12 \sum_{i=1}^N |\vev \sigma_i|^2 + U(\ve v_1,\dots \ve v_N)
\end{equation}
is dissipated by the dynamics. The following lemma summarizes the technical details obtained in \cite{ha2019}, and here adapted to the present case in which the $\alpha_i$'s are not necessarily zero.

\begin{lemma}[Energy dissipation]
\label{lemma:dissipazione}
Given $\{\ve v_i(t),{\vev \sigma}_i(t)\}_{i=1}^N$ solution to Eqs.~\eqref{eq:sigma-sistema}, let $E(t) = H(\ve v_1(t),\dots \ve v_N(t),\vev \sigma_1(t),\dots \vev \sigma_N(t))$. Then
\begin{enumerate}[(i)]
\item $\displaystyle 0\le  E(t) + \eta \int_0^t \sum_{i=1}^N \vev \sigma_i^2(s) \de s = E(0)$;
\item $|\vev \sigma_i(t)|\le \sqrt{2E(0)}$;
\item $\displaystyle \sup_{t\ge 0} \Big(\Big|\frac{\de^n\ve v_i}{\de t^n}(t)\Big| + \Big|\frac{\de^n\vev \sigma_i}{\de t^n}(t)\Big|\Big) < +\infty$ $\forall\, n\in \N$;
\item $\displaystyle \lim_{t\to+\infty}\vev \sigma_i(t) = 0$;
\item $\displaystyle \lim_{t\to+\infty}\dot{\vev \sigma_i}(t) = 0$;
\item Eq.~\eqref{winfty} holds for some $w_\infty \in [0,+\infty)$.
\end{enumerate}
\end{lemma}
\begin{proof}
We notice that $E(t)$ is a non-negative function and that, by explicit computation,
\[
\dot E(t) = - \eta \sum_{i=1}^N\vev \sigma_i^2(t)\,,
\]
so that assertion $(i)$ follows by integrating in time the above identity. In particular, $E(t) \le E(0)$ and therefore item $(ii)$ holds since $U\ge 0$. The assertion $(iii)$ for $n=1$ is obtained using the estimate $(ii)$ and that $|\ve v_i|=v$  in Eqs.~\eqref{eq:sigma-sistema}; the case $n>1$ then follows by repeatedly  differentiating Eqs.~\eqref{eq:sigma-sistema} with respect to the time $t$.

The assertion $(iv)$ follows from $(i)$ and $(iii)$, using the Barbalatt's lemma.\footnote{This lemma states that if $f(t)$ is an uniformly continuous function such that the limit $\lim_{t \to +\infty} \int_0^t f(t')\de t'$ exists and is finite, then $\lim_{t\to +\infty} f(t) =0$.} To prove assertion  $(v)$, we apply the Barbalatt's lemma to the function $\dot {\vev \sigma}_i(t)$, which has uniformly bounded derivative by $(iii)$ and whose time integral converges to $-\vev \sigma_i(0)$ in view of $(iv)$. It is worthwhile to notice that in the same way we could prove that all the time derivatives of higher order of $\ve v_i(t)$ and $\vev \sigma_i(t)$ vanishes as $t\to+\infty$. Finally, by the already proved items $(i)$ and $(iv)$, recalling \eqref{eq:sigma-potenziale} we have,
\[
\lim_{t\to +\infty} \frac {JN}{2v^2} |\ve w(t)|^2 = \frac {JN}{2v^2} |\ve w(0)|^2 - \frac 12 \sum_i |\vev \sigma_i(0)|^2 + \eta \int_0^{+\infty} \de s \, |\vev \sigma_i(s)|^2\,.
\]
Therefore, $|\ve w(t)|$ converges as $t\to+\infty$ to some limit $w_\infty$ (which is finite since $|\ve w(t)| \le v\sum_i n_i/N$). The lemma is thus proven.
\end{proof}

\begin{lemma}
\label{lemma:qL2}
Let $\{\ve v_i(t),{\vev \sigma}_i(t)\}_{i=1}^N$ be a solution to Eqs.~\eqref{eq:sigma-sistema} such that $w_\infty>0$, with $w_\infty$ as in item $(vi)$ of Lemma \ref{lemma:dissipazione}. Let $c>0$ and $\tau\ge 0$ be such that $|\ve w(t)|\ge c$ for any $t\ge \tau$ and define
\[
\ve q_i(t) = \ve v_i(t) \wedge \ve w(t)\,,\;\; t\in [0,+\infty)\,, \qquad \hat {\ve w}(t) = \frac{\ve w(t)}{|\ve w(t)|}\,,\;\; t\in [\tau,+\infty)\,.
\]
Then $\ve q_i(\cdot) \in L^2([0,+\infty))$ and $(\vev \sigma_i \cdot \hat {\ve w})(\cdot) \in L^1([\tau,+\infty))$.
\end{lemma}

\begin{proof}
From the definition of $\ve q_i(t)$, in view of Eqs.~\eqref{eq:sigma-sistema} we have,
\begin{align}
\label{eq:duederivate}
\dt (\vev \sigma_i \cdot \ve q_i)  & = \frac{J n_i}{v^2} |\ve q_i|^2 + \alpha_i \vev \sigma_i \cdot \ve w - \frac{\eta}{v^2} \vev \sigma_i \cdot \ve q_i + \vev \sigma_i \cdot \dot {\ve q}_i\,, \nonumber \\ \dt (\vev \sigma_i \cdot \ve w) & = - \frac{\alpha_i }{v^2}  \vev \sigma_i \cdot \ve q_i - \frac{\eta}{v^2} \vev \sigma_i \cdot \ve w +  \vev \sigma_i \cdot \dot{\ve w}\,, \nonumber \\ \dt (\vev \sigma_i \cdot \hat {\ve w}) & = - \frac{\alpha_i }{v^2|\ve w|}   \vev \sigma_i \cdot \ve q_i - \frac{\eta}{v^2} \vev \sigma_i \cdot \hat {\ve w} + \vev \sigma_i \cdot\dt \hat{\ve w}\,,
\end{align}
where we used the identity $(\vev \sigma_i \wedge \ve v_i) \cdot \ve w = \vev \sigma_i \cdot \ve q_i$ for deducing the second equation. We notice that the derivatives $\dot{\ve q}_i$, $\dot{\ve w}$, and $\frac{\de}{\de t}\hat {\ve w}$ can be bounded by a constant multiple of $\sum_i |\vev \sigma_i|$. Therefore, in view of assertion $(i)$ of Lemma \ref{lemma:dissipazione}, the last terms in the right-hand side of Eqs.~\eqref{eq:duederivate} belong to $L^1([0,+\infty)$. Now, from the first two identities in \eqref{eq:duederivate}, we have
\begin{align}
\label{eq:derivatona}
& \dt \left( \vev \sigma_i \cdot (\ve q_i +\frac {v^2\alpha_i}{\eta}\ve w) \right) = Jn_i |\ve q_i|^2 - \frac {\alpha_i^2 +\eta^2 }{\eta} \vev \sigma_i \cdot \ve q_i + \vev \sigma_i \cdot \left(\dot {\ve q}_i +\frac {v^2\alpha_i}{\eta} \dot {\ve w}\right) \nonumber \\ &  \quad = Jn_i \left( \ve q_i(t) -  \frac {\eta^2 + \alpha_i^2}{2\eta Jn_i} \vev \sigma_i \right)^2 - \frac {(\eta^2 + \alpha_i^2)^2}{(2\eta Jn_i)^2} \vev |\vev\sigma_i|^2+ \vev \sigma_i \cdot \left(\dot {\ve q}_i +\frac {v^2\alpha_i}{\eta} \dot {\ve w}\right)\,,
\end{align}
and we notice that the last three terms in the right-hand side belong to $L^1([0,+\infty)$. On the other hand, also the term in the left-hand side of Eq.~\eqref{eq:derivatona} belongs to $L^1([0,+\infty)$, since $\vev \sigma_i(t)\to 0$ for $t\to+\infty$, while $\ve q_i(t)$ and $\ve w(t)$ are uniformly bounded. Therefore,
\[
\ve q_i(\cdot) -  \frac {\eta^2 + \alpha_i^2}{2\eta Jn_i} \vev \sigma_i(\cdot) \in L^2([0,+\infty))
\]
and, since $\vev \sigma_i(\cdot) \in L^2([0,+\infty))$, we conclude that $\ve q_i (\cdot) \in L^2([0,+\infty))$.

Finally, consider the third identity in Eqs.~\eqref{eq:duederivate}. The term in the left-hand side belongs to $L^1([\tau,+\infty)$ since $\vev \sigma_i(t)\to 0$ as  $t\to+\infty$ and $|\hat{\ve w}(t)|=1$ Similarly, since $|\ve w(t)|>c$ for $t\ge \tau$, the functions $\vev \sigma_i \cdot \ve q_i/|\ve w|$ and $\vev \sigma_i \cdot \frac{\de}{\de t}\hat {\ve w}$ belong to $L^1([\tau,+\infty)$. In view of the aforementioned identity, we conclude that $\vev \sigma_i\cdot \hat{\ve w} \in L^1([\tau,+\infty))$.
\end{proof}

\noindent\textbf{Proof of Theorem \ref{thm:2}.}
By item $(vi)$ of Lemma \ref{lemma:dissipazione}, we have only to explore the dichotomy $|\ve w(t)|\to 0$ or $|\ve w(t)|\to w_\infty>0$. 

If $|\ve w(t)|\to 0$, Proposition 4.4 in \cite{ha2019} assures that $\{\ve v_i(t),\ve s_i(t)\}_{i=1}^N$ converges to a stationary incoherent solution of the system, in the case $\alpha_i=0$. The proof is based on the energy inequality, and therefore it works also in the case $\alpha_i \neq 0$. According to their reasoning, it is possible to prove that the function $t\mapsto |\vev \sigma_i(t)|^2\e^{\eta t/v^2}$ belongs to $L^1([0,+\infty))$, and this is sufficient to show that also $t\mapsto \vev \sigma_i(t) \wedge \ve v_i(t)$ is in $L^1([0,+\infty))$. For the details see \cite{ha2019}. This concludes the proof of the first part of the theorem.

If $|\ve w(t)| \to w_\infty>0$, since $\dot{\vev\sigma_i}(t) = \ve v_i(t) \wedge \ve w(t) \to 0$ (by item $(v)$ of Lemma \ref{lemma:dissipazione}) we obtain that also $\ve v_i(t) \wedge \hat {\ve w}(t) \to 0$. Hence,
\[
\lim_{t\to+\infty} | \ve v_i(t) \wedge \ve w(t)|^2 = v^2 - \lim_{t\to+\infty} (\ve v_i(t) \cdot \hat {\ve w}(t))^2 = 0\,,
\]
which is equivalent to
\[
\lim_{t\to+\infty} |\ve v_i(t) - v\hat {\ve w}(t)| = 0 \quad \text{or} \quad \lim_{t\to+\infty}|\ve v_i(t) + v\hat {\ve w}(t)|\to 0\,.
\]

Let now $c,\tau$ be as in Lemma \ref{lemma:qL2}. We claim that
\begin{equation}
\label{stimc}
\lim_{t\to+\infty} |\ve v_i(t) \pm v\hat {\ve w}(t)| = 0 \quad \Longrightarrow \quad  \ve v_i(\cdot) \pm v\hat {\ve w}(\cdot) \in L^2([\tau,+\infty))\,.
\end{equation}
Indeed, since $|\ve w(t)| \ge c$ for $t\ge \tau$, recalling $\ve q_i(t) = \ve v_i(t) \wedge \ve w(t)$, as $|\ve v_i(t) \cdot \hat {\ve w}(t)| \le v$ we have,
\[
0 \le v^2 - v|\ve v_i(t) \cdot \hat {\ve w}(t)| \le v^2 - (\ve v_i(t) \cdot \hat {\ve w}(t))^2 \le \frac{1}{c^2}|\ve q_i(t)|^2 \qquad \forall\, t\ge \tau\,.
\]
Therefore, in view of Lemma \ref{lemma:qL2}, the function $t\mapsto v^2 - v|\ve v_i(t) \cdot \hat {\ve w}(t)|$ belongs to $L^1([\tau,+\infty))$. The claim \eqref{stimc} then follows by noticing that $v^2 - v|\ve v_i(t) \cdot \hat {\ve w}(t)| = \frac 12 |\ve v_i - v\hat {\ve w}|^2$ if $\ve v_i \cdot \hat {\ve w} \ge 0$, while $v^2 - v|\ve v_i(t) \cdot \hat {\ve w}(t)| = \frac 12 |\ve v_i + v\hat {\ve w}|^2$ if $\ve v_i \cdot \hat {\ve w} \le 0$.

We now prove the convergence of $\hat {\ve w}(t)$ as $t\to+\infty$, from which that of $\ve v_i(t)$ will follow. Let $I^+$ be the set of indexes $i$ such that $|\ve v_i(t) -v \hat {\ve w}(t)|\to 0$, $I^-$ be the set of indexes $i$ such that $|\ve v_i(t) + v\hat {\ve w}(t)|\to 0$, and define
\[
\ve u(t) = \frac 1N \left( \sum_{i\in I^+} \ve v_i(t) -\sum_{i\in I^-} \ve v_i (t)\right).
\]
Its time derivative can be written as
\[
\dot {\ve u} = \frac 1N \sum_{i \in I^+} \vev \sigma_i \wedge (\ve v_i - v\hat {\ve w}) - \frac 1N \sum_{i \in I^-} \vev \sigma_i \wedge (\ve v_i + v\hat {\ve w}) +v \sum_{i} \vev \sigma_i \cdot \hat {\ve w}\,.
\]
Integrating in time from $\tau$ to $t$, using that $\vev \sigma_i \in L^2([0,+\infty))$ and Eq.~\eqref{stimc}, from the Cauchy-Schwarz inequality it follows that it exists the limit 
\[
\lim_{t\to +\infty} \ve u(t) = \ve u_{\infty}\,.
\]
Finally, since
\[
\ve u - v\hat {\ve w}= \frac 1N \sum_{i\in I^+} (\ve v_i - v \hat {\ve w}) - \frac 1N \sum_{i\in I^-} (\ve v_i + v\hat {\ve w})
\]
then $\hat {\ve w} \to v^{-1} \ve u_{\infty} \in S_1^2$.
\hfill\qed

\begin{corollary}
\label{cor:1}
Under the hypothesis of Theorem \ref{thm:2}, let $E(t)$ be the energy as defined in Lemma \ref{lemma:dissipazione} and set $m= \sum_i n_i$, $n_-=\min_i n_i$. If the initial datum satisfies
\begin{equation}
\label{condE}
E(0) < \frac{Jm^2}{2N}
\end{equation}
then the solution tends to an aligned stationary solution. Moreover, when the more restrictive condition
\begin{equation}
\label{condE0}
E(0) < \frac{2J n_- (m-n_-)}N
\end{equation}
is fulfilled, the solution tends to a flocking stationary solution.
\end{corollary}

\begin{proof}
The content of this corollary is already present in \cite{ha2019}, but for the  sake of completeness we give here the details of the proof. From the expression of the potential energy in Eq.~\eqref{eq:sigma-potenziale}, the condition \eqref{condE} reads,
\[
\frac {v^2}{JN} \sum_{i=0}^N \left|\vev \sigma_i(0) \right|^2< |\ve w(0)|^2\,,
\]
which combined with the dissipation inequality $E(t) \le E(0)$ implies that
\[
|\ve w(t)|^2 \ge |\ve w(0)|^2 - \frac{v^2}{JN} \sum_{i=0}^N |\vev \sigma_i(0)|^2\,.
\]
Therefore, if the initial data satisfy the estimate \eqref{condE} then  the right-hand side is positive, hence $|\ve w(t)|$ cannot vanishes as $t\to+\infty$, and this implies the convergence to a aligned stationary solution for what stated before.

Recalling \eqref{eq:sigma-potenziale}, we observe that the potential energy $U$ computed on a aligned stationary solution with at least two particles in two opposite poles of $S_v^2$ is given by
\[
\begin{split}
U & = \frac{J}{2N}\left(\sum_i n_i\right)^2 - \frac {J}{2N} \left(\sum_{i\in I^+}n_i - \sum_{i\in I^-} n_i \right)^2 = \frac{J}{2N}\left[m^2 - \left(m-2\sum_{i\in I^+}n_i \right)^2\right] \\ & \ge \frac{J}{2N} \min_{i=1\,\ldots,N}\left[m^2 - \left(m-2n_i \right)^2\right] = \frac{2J}{N} \min_{i=1\,\ldots,N} n_i(m-n_i)  = \frac{2J n_- (m-n_-)}N\,,
\end{split}
\]
where, concerning the last identity, it is sufficient to notice that if $k$ is such that $n_-=n_k$ then, for any $i=1,\ldots,N$,
\[
n_i(m-n_i) = n_in_k + n_i \sum_{j\ne i,k}n_j \ge n_i n_k + n_k \sum_{j\ne i,k} n_j  = n_k\sum_{j\ne k} n_j = n_k(m-n_k)\,.
\]
Therefore, when $E(0)$ satisfies the estimate \eqref{condE} the asymptotic state is necessarily a flocking stationary solution. The theorem is thus proved.
\end{proof}	

\section{The free space system with diffusion}
\label{sec:vs.bagno}

In this section, we consider the system in presence of thermal bath and we are mainly interested in its stochastic mean-field limit, which is realized by assuming constant communication rates $n_{ij} =1$ and taking the limit $N\to +\infty$. The equations of motion then read,\footnote{Eqs.~\eqref{eq:vsstoc} are defined also if $|\ve v_i(0)|  \ne v$, but the correct interpretation of the model requires $|\ve v_i(0)|=v$. It is worthwhile to notice that some properties of the dynamics depend on this assumption, see, e.g., Lemma \ref{lemma:dissipazionestoc}.}
\begin{equation}
\label{eq:vsstoc}
\left\{
\begin{aligned}
&\de \ve v_i = \ve s_i \wedge \ve v_i\de t\,,\\
&\de\ve s_i = \ve v_i \wedge \left( \frac J{v^2} \ve w \de t - \frac \eta{v^2} \de \ve v_i  + \frac{\sqrt{2\nu}}{v}\de \ve B_i\right), 
\\ & \ve w =  \frac 1N \sum_{j=1}^N  \ve v_j\,.
\end{aligned}
\right.
\end{equation}
We rewrite $\ve v_i \wedge \de \ve B_i = \Omega (\ve v_i) \de \ve B_i$ where, given a vector $\ve u = (u_1, u_2, u_3)$, $\Omega (\ve u)$ denotes the skew-symmetric matrix 
\begin{equation}
\label{eq:omega}
\Omega (\ve u) = \begin{pmatrix} 0 & u_3 & -u_2 \\  -u_3 & 0 & u_1 \\ u_2 & -u_1 & 0 \end{pmatrix},
\end{equation}
and recall that $\ve v_i \wedge (\ve s_i \wedge \ve v_i) = |\ve v_i|^2 P_i^\perp \ve s_i$, where $P^\perp_i = \Id - \hat {\ve v}_i \otimes \hat {\ve v}_i$. Therefore, Eqs.~\eqref{eq:vsstoc} can be reshaped in the following form,
\begin{equation}
\label{eq:vsstoc1}
\left\{
\begin{aligned}
&\de \ve v_i = \ve s_i \wedge \ve v_i\de t\,,\\
&\de\ve s_i =  \frac J{v^2} \ve v_i \wedge \ve w \de t - \frac \eta{v^2}  |\ve v_i|^2 P_i^\perp \ve s_i \de t + \frac{\sqrt{2\nu}}{v} \Omega (\ve v_i) \de \ve B_i\,,\\ 
& \ve w =  \frac 1N \sum_{j=1}^N  \ve v_j\,.
\end{aligned}
\right.
\end{equation}
The natural phase space of Eqs.~\eqref{eq:vsstoc1} is given by $(\R^6_0)^N$, where $\R^6_0 := (\R^3\setminus\{0\}) \times \R^3$. Moreover, as the right-hand side of Eqs.~\eqref{eq:vsstoc1} is locally Lipschitz, local (in time) existence of strong solutions and pathwise uniqueness follow by standard theory on stochastic differential equations. 

As the right-hand side of Eqs.~\eqref{eq:vsstoc1} is orthogonal to both $\ve v_i$ and $\ve s_i$, for any initial datum in $(\R^6_0)^N$ we have,
\begin{equation}
\label{vvs}
\de |\ve v_i(t)|^2 = 0\,, \quad \de (\ve v_i(t) \cdot \ve s_i(t)) = 0 \qquad \forall\, i = 1,\ldots,N\,.	
\end{equation} 
Since we always assume that $|\ve v_i(0)|=v$ a.e., in view of the first identity in \eqref{vvs} the velocities belong, for any positive time, to the spherical surface of radius $v$, hereafter denoted by $S_v^2$. Otherwise stated, the motion actually takes place a.s.\ on the hyper-surface $(S_v^2\times \R^3)^N$. 

\begin{lemma}[Energy inequality] 
\label{lemma:dissipazionestoc}
Let $E$ be the energy of the system,
\[
E := \frac 12 \sum_{i=1}^N |\ve s_i|^2 + U\,, \qquad U := \frac J{4N v^2}\sum_{ij} (\ve v_i -\ve v_j)^2\,.
\]
Then (assuming $|\ve v_i(0)|=v$),
\begin{equation}
\label{stimE}
\E (E(t)) \le \E( E(0)) + \nu Nt\,.
\end{equation}
\end{lemma}

\begin{proof}
Since $|\ve v_i(0)|=v$, by \eqref{vvs} we have
\[
U = \frac {JN}{2} - \frac {JN}{2v^2} |\ve w|^2 \quad \text{a.s.}
\]
Therefore, by It\^o formula and using that $\Tr \Omega^t(\ve u)\Omega(\ve u)= 2|\ve u|^2$,
\[
\de E = - \eta \sum_{i=1}^N |P_i^\perp \ve s_i|^2\de t + \nu N\de t + \frac{\sqrt{2\nu}}v  \sum_{i=1}^N \ve s_i \cdot \Omega (\ve v_i) \de \ve B_i\,.
\]
Integrating on $[0,t]$ and then taking the expected value, we obtain
\[
\E (E(t)) + \eta \int_0^t \sum_{i=1}^N \E (|P^\perp_i\ve s_i|^2) = \E( E(0)) + \nu Nt\,,
\]
which implies \eqref{stimE}.
\end{proof}

Since $\E(|\ve s_i(t)|^2) \le \E(2E(t))$, if we assume $\E (E(0))<+\infty$ (e.g., in the case of deterministic initial data) then $\E(|\ve s_i(t)|^2) < +\infty$. By standard theory of stochastic differential equations, this guarantees global existence and uniqueness of solutions for  Eqs.~\eqref{eq:vsstoc1}. 

As already noticed, the model is designed to evolve initial data with not only $|\ve v_i(0)| = v$ but also $\ve v_i(0) \cdot \ve s_i(0)=0$ for all $i$. If we restrict the analysis to such initial data, in view of the second identity in \eqref{vvs}, the spin variables $\ve s_i$ are a.s.\ confined to the tangent space $T_{\ve v_i}(S_v^2)$ for all times. Nevertheless, since not relevant in the present discussion, throughout this section we do not restrict the analysis to such initial data. On the contrary, in the next section, when discussing the stationary distributions in the mean-field limit, we will focus on the stationary measures which are supported on the tangent bundle. 

\smallskip
We now discuss the limit process describing the particle system in the mean-field limit $N\to +\infty$. In view of the physical interpretation of the model, it is reasonable to assume that the joint law of the initial conditions $\{(\ve v_i(0),\ve s_i(0))\}_{i=1}^N$ is symmetric with respect to particle permutations. Indeed, we make the stronger assumption that the variables $\{(\ve v_i(0),\ve s_i(0))\}_{i=1}^N$ are i.i.d., with a common density $f_0(\ve v,\ve s)$ with respect to the reference measure $\de\Sigma_{\ve v} \de \ve s$ in the single particle phase $S_v^2\times \R^3$, where $\de\Sigma_{\ve v}$ denotes the surface measure on $S_v^2$. 

As usual in the framework of mean-field models, despite correlations may occur as time goes by due to the interactions, we do expect propagation of chaos in the limit $N\to +\infty$. As discussed in \cite{S}, in such limit the collective behavior of the particles is well described by the auxiliary process $(\de \ve v(t),\de \ve s(t))$ on $S_v^2\times \R^3$ defined by the solution of the following system of stochastic integro-differential equations,
\begin{equation}
\label{eq:ausiliario}
\left\{
\begin{aligned}
&\de {\ve v} = {\ve s} \wedge
{\ve v}\de t\,,\\
&\de {\ve s} = {\ve v} \wedge \left( \frac J{v^2}{\ve w}_f\de t -
\frac {\eta}{v^2} \de {\ve v}  + \frac {\sqrt{2\nu}}v \de \ve B\right), \\
&\ve w_f(t)  =  \int_{S_v^2\times \R^3} \ve v f_t(\ve v,\ve s,t) \de \Sigma_{\ve v} \de \ve s\,,
\end{aligned}
\right.
\end{equation}
where $\ve B$ is a standard Brownian motion in $\R^3$, the initial datum satisfies $|\ve v(0)|=v$ a.s., and $f_t$ is the density of the law of $(\ve v(t), \ve s(t))$. This system is well posed since also in this case $\de |\ve v|^2=0$,  so that $f_t$ is supported on $S_v^2 \times \R^3$ (provided this is true a time $t=0$).\footnote{To be more precise, one first considers the process $(\ve v(t),\ve s(t))$ in the whole $\R^6$, solution to \eqref{eq:ausiliario} with the last equation replaced by $\ve w_\mu(t)  =  \int_{\R^6} \mu_t(\de \ve v, \de \ve s) \ve v$ for arbitrary measure $\mu_t(\de\ve v,\de \ve s)$. By the conservation law $\de |\ve v|^2=0$ it follows that if $\mu_0$ is supported on $S_v^2 \times \R^3$ then this is true also for $\mu_t$. Finally, by standard regularity properties of diffusion processes, if $\mu_0(\de\ve v,\de \ve s) = f_0(\de\ve v,\de \ve s)\, \de \Sigma_{\ve v} \de \ve s$ then $\mu_t(\de\ve v,\de \ve s) = f_t(\de\ve v,\de \ve s)\, \de \Sigma_{\ve v} \de \ve s$ with $f_t$ solution to \eqref{fpd}.}

The mean-field equation describing the dynamics of the system in the limit $N\to +\infty$ is the equation governing the evolution of the law $f=f_t$. It is given by the following non-linear Fokker-Planck equation,
\begin{equation}
\label{fpf}
\!\!\left\{
\begin{aligned}
&\pa_t f + \dive_{\ve v} (\ve s \wedge \ve v  f) + \frac J{v^2} \nabla_{\ve s} \cdot (\ve v \wedge \ve w_f f)  = \eta  \nabla_{\ve s} \cdot (P^\perp \ve s f) + \nu \Tr (P^\perp D^2_{\ve s} f)\,,\\
&\ve w_f = \int_{S_v^2\times \R^3} \ve v f \de \Sigma_{\ve v} \de \ve s\,,\\
\end{aligned}
\right.
\end{equation}
where $\dive_{\ve v}(\cdot)$ is the divergence operator on the surface $S_v^2$ and $\Tr (P^\perp D_{\ve s}^2 f) = \Delta_{\ve s} f - \sum_{i,j=1}^3 (\hat {\ve v}_i \cdot  \hat {\ve v}_j) D_{\ve s_i \ve s_j} f$. For non-smooth densities, the above equation has to be interpreted via its weak formulation: for any $\varphi\in C^{\infty}_0(S_v^2\times \R^3)$,
\begin{eqnarray}
\label{fpd}
\pa_t \int \varphi f_t & = & \int\Big( P^\perp\nabla_{\ve v}\varphi \cdot (\ve s\wedge \ve v)+ \frac J{v^2}  \nabla_{\ve s} \varphi \cdot (\ve v\wedge \ve w_f) \nonumber \\ && \qquad +\,  \eta \nabla_{\ve s} \varphi \cdot P^\perp \ve s + \nu \Tr (P^\perp  D^2_{\ve s}  \varphi)\Big) f_t\,.
\end{eqnarray}
We remark that in \eqref{fpd} $P^\perp\nabla_{\ve v}$ is the gradient operator on the surface $S_v^2$ so that, by integration by parts,
\begin{equation}
\label{ipv}
\int_{S_v^2} \varphi \dive_{\ve v} (\ve s \wedge \ve v  f) \de \Sigma_{\ve v} = - \int_{S_v^2} P^\perp\nabla_{\ve v} \varphi  \cdot (\ve s \wedge \ve v  f) \de \Sigma_{\ve v}\,.
\end{equation}

Next, we note that, defining
\[
H_f(\ve v,\ve s) = \frac 12 |\ve s|^2 + \frac J{2v^2} \int_{S_v^2\times \R^3} (\ve v-\ve v')^2 f_t(\ve v',\ve s') \de \Sigma_{\ve v'} \de \ve s'
\]
we have $\nabla_{\ve v} H_f = J(\ve v - \ve w_f)/v^2$ so that 
\[
\left\{
\begin{aligned} 
& \de \ve v = -\ve v \wedge \nabla_{\ve s} H_f\,, \\ & \de {\ve s} = -\ve {\ve v} \wedge (\nabla_{\ve v} H_f  - \eta \de {\ve v}/v^2  + \sqrt{2\nu}/v \de \ve B)\,. 
\end{aligned}
\right.
\]
Proceeding as above, we can prove that
\[
\E (H_{f_t}) + \nu \int_0^t \E (|P^\perp \ve s|^2) = \E (H_{f_0}) + \nu t\,,
\]
where $P^\perp = \Id - \hat {\ve v} \otimes \hat {\ve v}$. In particular, $\E(|\ve s(t)|^2)< +\infty$ thus obtaining also in this case global existence and pathwise uniqueness of the solution for any initial distribution $f_0$.

The precise relation between the auxiliary process defined via Eqs.~\eqref{eq:ausiliario} and the particle system can be detailed as it follows. Let $\{(\bar{\ve v}_i(t),\bar{\ve s}_i(t))\}_{i=1}^N$ be the solutions to Eqs.~\eqref{eq:ausiliario} obtained choosing $\ve B=\ve B_i$, with $\{\ve B_i\}_{i=1}^N$ as in Eqs.~\eqref{eq:vsstoc}, and $\{(\bar{\ve v}_i(0),\bar{\ve s}_i(0))\}_{i=1}^N = \{(\ve v_i(0),\ve s_i(0))\}_{i=1}^N$. We note that the distribution $f_t$ of $(\bar {\ve v}_i (t), \bar {\ve s}_i (t))$ is independent of the index $i$ since the joint law of the initial conditions $\{(\ve v_i (0), \ve s_i (0))\}_{i=1}^N$ is assumed to be       symmetric with respect to permutations of particle indexes. From the theory in \cite{S} it can be proved that for each $T>0$ there is a constant $C>0$ such that, for any $i=1,\ldots,N$,
\begin{equation}
\label{convq}
\E \left(\sup_{t\in [0,T]} \left(|\ve v _i (t) - \bar{\ve v} _i (t)| + |\ve s_i (t) - \bar {\ve s} _i (t)| \right)  \right)\le \frac CN\,,
\end{equation}
from which the convergence as $N\to+\infty$ of the law of $(\ve v_i (t), \ve s_i (t))$ (and of the empirical measure) to $f_t$ follows from classical estimates. The proof of \eqref{convq} requires some technicalities (regularization of the diffusion and drift coefficients, a priori estimates, removal of the cut-off), and can be done by arguing as in \cite{jac2012}, where the analogous analysis is performed for a continuous-time, stochastic version of the Vicsek model. We omit the details.

\section{Free energy functional and equilibrium solution}

The free energy functional associated to the evolution equation \eqref{fpf} is the functional $\Cal F$ defined as
\[
\begin{split}
\Cal F(f) & = \int_{S_v^2\times \R^3} f(\ve v, \ve s) \log f(\ve v, \ve s) \de \Sigma_{\ve v} \de \ve s + \beta \int_{S_v^2 \times \R^3} \frac 12 |\ve s|^2 f(\ve v, \ve s) \de \Sigma_{\ve v} \de \ve s \\ & \qquad + \frac{\beta J}{4v^2} \int_{S_v^2\times \R^3} \de \Sigma_{\ve v} \de \ve s \int_{S_v^2\times \R^3} \de \Sigma_{\ve v'} \de \ve s' (\ve v-\ve v')^2 f(\ve v, \ve s) f(\ve v' ,\ve s')\,,
\end{split}
\]
where $\beta :=  \eta/\nu$. For our purposes, it is useful to recast $\Cal F$ in the following form,
\begin{equation}
\label{F}
\Cal F(f) =\int_{S_v^2\times \R^3} f \log f \de \Sigma_{\ve v} \de \ve s + \frac \beta 2 \int_{S_v^2 \times \R^3} |\ve s|^2 f \de \Sigma_{\ve v} \de \ve s - \frac{\beta J}{2v^2} |\ve w_f|^2 + \frac{\beta Jv^2}2\,.
\end{equation}

\begin{remark}
\label{rem:onsager}
Consider densities corresponding to independent distributions of velocity and spin. Each of such densities $f$ can be written as a product, $f(\ve v, \ve s) = f_*(\ve s)  h(\hat{\ve v}) $, with  $\hat{\ve v} = v^{-1}\ve v$ and $h$ a probability density on the unit sphere $S_1^2$. Then $\Cal F(f) = \beta Jv^2 \Cal G(h) + \mathrm{const}$, where $\Cal G$ is the \textit{Onsager free energy functional},
\[
\Cal G(h) = \int_{S_1^2} \left[ kT h(\hat{\ve v}) \log h(\hat{\ve v})  + \frac 12 h(\hat{\ve v})  \int_{S_1^2} U_d(\hat{\ve v},\hat{\ve v}') h(\hat{\ve v}') \de \Sigma_{\hat{\ve v}'} \right]\de \Sigma_{\hat{\ve v}}\,,
\]
with $kT=(\beta J)^{-1}$ and dipolar interaction potential $U_d(\hat{\ve v},\hat{\ve v}') = - \hat{\ve v} \cdot \hat{\ve v}'$. The critical points of this functional have been studied in \cite{FS}. This analysis has been extended to arbitrary dimension in \cite{FL}, where the authors study the Smoluchowski equation on the $n$-sphere with dipolar potential, which is derived as the (spatial-homogeneous) kinetic mean-field equation for a continuous-time stochastic variant of the Vicsek model.
\end{remark}

\begin{theorem}[Dissipation of the free energy] 
\label{thm:p1}
If $f_t$ is a regular solution to Eq.~\eqref{fpf} then
\begin{equation}
\label{fep}
\frac{\de}{\de t}\Cal F(f_t) = - \nu \int_{S_v^2\times \R^3} \frac 1{f_t} |P^\perp (\nabla_{\ve s} f_t + \beta \ve sf_t)|^2 \de \Sigma_{\ve v} \de \ve s
\end{equation}
(recall $\beta =  \eta/\nu$). Moreover, a regular stationary solution has the form,
\[
f(\ve v, \ve s) = g(\ve v, \ve v \cdot \ve s) \e^{-\frac \beta 2 |P^\perp \ve s|^2}\,,
\]
for some regular function $g\colon \R^3\times \R \to [0,+\infty)$.
\end{theorem}

\begin{proof}
In view of \eqref{F} and recalling the definition of $\ve w_f$ in \eqref{fpf}, we have,
\[
\frac{\de}{\de t}\Cal F(f_t)  = D_1 + D_2 + D_3\,,
\]
with
\[
\begin{split}
D_1 & = \int_{S_v^2 \times \R^3} (1+\log f)\pa_t f \de \Sigma_{\ve v} \de \ve s \,, \quad D_2 = \frac \beta 2 \int_{S_v^2 \times \R^3} |\ve s|^2 \pa_t f \de \Sigma_{\ve v} \de \ve s\,, \\
D_3 & = - \frac{\beta J}{v^2} \ve w_f \cdot \int_{S_v^2 \times \R^3}\ve v \pa_t f \de \Sigma_{\ve v} \de \ve s\,,
\end{split}
\]
where we have shortened $f_t=f$. To compute these terms we use \eqref{fpf} and, recalling also \eqref{ipv}, we perform integration by parts with respect to both the variables $\ve v$ and $\ve s$. We thus obtain, after some straightforward computations,
\[
\begin{split}
D_1 & = \int_{S_v^2 \times \R^3} \de \Sigma_{\ve v} \de \ve s\, \Big[ P^\perp\nabla_{\ve v} f \cdot (\ve s \wedge \ve v) + \frac{J}{v^2} \nabla_{\ve s}f \cdot (\ve v \wedge \ve w_f) \\ & \qquad - \eta\nabla_{\ve s}f \cdot P^\perp \ve s  - \frac{\nu}f \nabla_{\ve s}f \cdot P^\perp \nabla_{\ve s}f \Big]\,,
\\ D_2 & = \int_{S_v^2 \times \R^3} \de \Sigma_{\ve v} \de \ve s\, \beta \Big[-\frac 12|\ve s|^2 \dive_{\ve v}(\ve s \wedge \ve v  f) + \frac{J}{v^2}\, \ve s \cdot  (\ve vf \wedge \ve w_f) \\ & \qquad  - \eta f \, \ve s \cdot P^\perp \ve s  - \nu \, \ve s \cdot P^\perp \nabla_{\ve s}f \Big]\,, 
\\ D_3 & = \frac{\beta J}{v^2} \ve w_f \cdot \int_{S_v^2 \times \R^3} \de \Sigma_{\ve v} \de \ve s\, \ve v \, \Big[\dive_{\ve v}(\ve s \wedge \ve vf) +  \frac{J}{v^2} \nabla_{\ve s}f \cdot  (\ve v \wedge \ve w_f) \\ & \qquad\quad  +\eta  \nabla_{\ve s} \cdot (P^\perp \ve s f) + \nu \Tr (P^\perp D^2_{\ve s} f) \Big]\,.
\end{split}
\]
The first two terms in the expression of $D_1$, the first term in the expression of $D_2$ and the last three terms in the expression of $D_3$ are vanishing boundary terms. The sum of the second term in the expression of $D_2$ and the first term in the expression of $D_3$ is zero, as
\[
\begin{split}
\int_{S_v^2} \de \Sigma_{\ve v}\, (\ve w_f \cdot \ve v) \dive_{\ve v} (\ve s \wedge \ve vf) & = - \int_{S_v^2} \de \Sigma_{\ve v}\, P^\perp\nabla_{\ve v} (\ve w_f \cdot \ve v) \cdot (\ve s \wedge \ve vf) \\ & = \int_{S_v^2} \de \Sigma_{\ve v}\, \ve w_f \cdot (\ve s \wedge \ve vf)\,.
\end{split}
\]
Collecting together all the non zero terms in $D_1+D_2+D_3$ and using $\beta\nu=\eta$ we obtain Eq.~\eqref{fep},
\[
\begin{split}
\frac{\de}{\de t}\Cal F(f)  & = -  \int_{S_v^2 \times \R^3} \de \Sigma_{\ve v} \de \ve s\, \Big[\eta\nabla_{\ve s}f \cdot P^\perp \ve s  + \frac{\nu}f \nabla_{\ve s}f \cdot P^\perp \nabla_{\ve s}f  \\ & \qquad + \beta\eta f \, \ve s \cdot P^\perp \ve s  + \beta\nu \, \ve s \cdot P^\perp \nabla_{\ve s}f \Big] \\ & = - \nu \int_{S_v^2 \times \R^3} \de \Sigma_{\ve v} \de \ve s\, \Big[\frac 1f |P^\perp \nabla_{\ve s}f|^2  + \beta^2 f \, |P^\perp \ve s|^2  + 2 \beta P^\perp\nabla_{\ve s}f \cdot P^\perp \ve s \Big] \\ & = - \nu \int_{S_v^2\times \R^3} \frac 1{f} |P^\perp (\nabla_{\ve s} f + \beta \ve s f_t)|^2\,.
\end{split}
\]

Suppose now that $f=f(\ve v, \ve s)$ is a regular stationary solution to Eq.~\eqref{fpf}. In view of \eqref{fep} $f$ satisfies $P^\perp (\nabla_{\ve s} f + \beta \ve s)=0$. Looking for $f$ in the form $f = h \e^{-\frac \beta 2 |P^\perp \ve s|^2}$, we deduce that the unknown $h$ satisfies $P^\perp \nabla_{\ve s} h =0$, i.e., $\nabla_{P^\perp \ve s} h = 0$, which implies $h(\ve v, \ve s) = g(\ve v, \ve v \cdot \ve s)$ for some $g\colon \R^3\times \R \to [0,+\infty)$.
\end{proof}

In view of the conservation of $\ve v \cdot \ve s$, we can look for stationary measures  supported in $\{(\ve v,\ve s) \in S_v^2\times \R^3 \colon \ve v \cdot \ve s = \alpha\}$, with $\alpha\in \R$. In particular, as before, we consider the case $\alpha=0$. In this case the support of the measure is the tangent bundle $TS_v^2 $. We seek this measure as a (generalized) solution to \eqref{fpf} of the form
\[
\mu(\de \ve v,\de \ve s) = g(\ve v) \e^{-\frac \beta 2 |P^\perp \ve s|^2} \delta (\ve v \cdot \ve s)\,.
\]
Inserting this expression in the weak formulation \eqref{fpd} and restricting the integration to the tangent bundle due to the delta function, we obtain, for any $\varphi\in C^{\infty}_0(S_v^2\times \R^3)$,
\[
\begin{split}
& \int_{S_v^2} \de \Sigma_{\ve v} \int_{T_{\ve v}S_v^2} \de \ve s  \left(P^\perp\nabla_{\ve v}\varphi \cdot (\ve s\wedge \ve v) + \frac J{v^2}  \nabla_{\ve s} \varphi \cdot (\ve v\wedge \ve w_\mu) \right) g(\ve v) \e^{-\frac \beta 2 |P^\perp \ve s|^2}  \\ & \qquad \qquad = - \int_{S_v^2} \de \Sigma_{\ve v} \int_{T_{\ve v}S_v^2} \de \ve s  \left(\eta \nabla_{\ve s} \varphi \cdot P^\perp \ve s + \nu \Tr (P^\perp  D^2_{\ve s}  \varphi)\right) g(\ve v) \e^{-\frac \beta 2 |P^\perp \ve s|^2}\,.
\end{split}
\]
We note that $T_{\ve v}S_v^2 = \{\ve s\colon \ve v \cdot \ve s=0\}$ and the operator $P^\perp$ drops out the derivative in the direction of $\ve v$, then we can integrate by part in $T_{\ve v}S_v^2$ as usual in $\R^2$. Therefore, integrating by parts with respect to both the variables $\ve v$ and $\ve s$ we have,
\[
\begin{split}
& \int_{S_v^2} \de \Sigma_{\ve v} \int_{T_{\ve v}S_v^2} \de \ve s  \left(\dive_{\ve v}(\ve s\wedge \ve v g)\e^{-\frac \beta 2 |P^\perp \ve s|^2} + \frac J{v^2}  \dive_{\ve s} (\ve v\wedge \ve w_\mu \e^{-\frac \beta 2 |P^\perp \ve s|^2})g \right)\varphi  \\ & \quad = \int_{S_v^2} \de \Sigma_{\ve v} \int_{T_{\ve v}S_v^2} \de \ve s  \left(- \eta \nabla_{\ve s} \cdot (P^\perp \ve s\, \e^{-\frac \beta 2 |P^\perp \ve s|^2}) + \nu \Tr (P^\perp  D^2_{\ve s} \e^{-\frac \beta 2 |P^\perp \ve s|^2} )\right) g(\ve v) \varphi\,.
\end{split}
\]
The right-hand side is zero by direct computation and recalling that $\beta=\mu/\nu$. Therefore, by the arbitrariness of $\varphi$, we conclude that $g$ has to be such that
\[
\dive_{\ve v}(\ve s\wedge \ve v g)\e^{-\frac \beta 2 |P^\perp \ve s|^2} + \frac J{v^2}  \dive_{\ve s} (\ve v\wedge \ve w_\mu \e^{-\frac \beta 2 |P^\perp \ve s|^2})g =0\,.
\]
By direct calculation (e.g., using angular coordinates) it is easily verified that $\dive_{\ve v}(\ve s\wedge \ve v g) = \nabla_{\ve v} g \cdot (\ve s \wedge \ve v)$. Similarly, $\dive_{\ve s} (\ve v\wedge \ve w_\mu \e^{-\frac \beta 2 |P^\perp \ve s|^2}) = - \beta P^\perp \ve s \cdot (\ve v\wedge \ve w_\mu) \e^{-\frac \beta 2 |P^\perp \ve s|^2}$. Therefore $g=g(\ve v)$ is  solution to
\[
\ve s\cdot (\ve v\wedge \nabla_{\ve v} g) = \frac{\beta J}{v^2} \ve s \cdot (\ve v\wedge \ve w_\mu) \qquad \forall\, (\ve v,\ve s) \in TS_v^2\,,
\]
whence $\nabla_{\ve v} g = \frac{\beta J}{v^2} \ve w_\mu$, so that
\[
g(\ve v) = \frac{1}{Z}\exp\Big( \frac{\beta J}{v^2} \ve w_\mu \cdot \ve v\Big)\,,
\]
where $Z$ is the normalization constant,
\[
Z =  \int_{S_v^2} \de \Sigma_{\ve v} \int_{T_{\ve v}S_v^2} \de \ve s \exp\Big(\frac{\beta J}{v^2} \ve w_\mu \cdot \ve v -\frac \beta 2 |\ve s|^2 \Big)\,,
\]
and $\ve w_\mu =: \ve w$ has to satisfy the following relation,
\begin{equation}
\label{cm}
\ve w = \frac{\int_{S_v^2} \de \Sigma_{\ve v} \ve v \exp\Big( \frac{\beta J}{v^2} \ve w \cdot \ve v\Big)}{\int_{S_v^2} \de \Sigma_{\ve v} \exp\Big( \frac{\beta J}{v^2} \ve w \cdot \ve v\Big)}\,,
\end{equation}
which is a self-compatibility condition for the existence for $g$. Eq.~\eqref{cm} and its non-zero solutions have been already appeared in the literature, since they arise when looking for the critical points of the Onsager free energy functional, whose relation with our system has been discussed in Remark \ref{rem:onsager}. Nevertheless, also to make the presentation more clear and self-contained, we prefer to discuss its solutions in some detail.

Eq.~\eqref{cm} always admits the solution $\ve w = 0$. We search for non zero solutions $\ve w = \gamma v\ve e$, with $\gamma>0$ and $\ve e\in S_1^2$. Using spherical coordinates $\vartheta$, $\varphi$ such that $v\cos \vartheta = \ve v \cdot \ve e$ and $\varphi$ is the angle around the $\ve e$ axes, we can calculate explicitly the integrals,
$$
\begin{aligned}
\int_{S_v^2} \de \Sigma_{\ve v} \, \e^{\beta J \ve w \cdot \ve v/v^2} & = 2\pi \int_0^\pi \de \vartheta \sin \vartheta\,  \e^{\beta J \gamma \cos \vartheta}  = \frac{4\pi }{\beta J \gamma} \sinh (\beta J \gamma)\\ \int_{S_v^2} \de \Sigma_{\ve v}\,  \ve v\, \e^{\beta J \ve w \cdot \ve v/v^2}  & =  2\pi \ve e \int_0^\pi \de \vartheta \sin \vartheta \cos \vartheta \, \e^{\beta J \gamma \cos \vartheta}
\\ & = 4\pi \ve e \left(\frac 1{\beta J \gamma} \cosh ({\beta J \gamma}) -  \frac 1{\beta^2 J^2 \gamma^2} \sinh (\beta J \gamma) \right).
\end{aligned}
$$
Defining $\xi = \beta J \gamma$, we finally obtain the equation 
\[
\xi = \beta J \left( \frac 1{\tanh \xi} - \frac 1\xi\right).
\]

\begin{lemma}
\label{lem:h}
The function $x \to h(x) = \frac 1{\tanh x} - \frac 1x$, defined for $x\in [0,+\infty)$ setting $h(0) =0$, is a strictly monotone increasing and concave function with $h'(0^+) = 3$.
\end{lemma}

\begin{proof} 
The function $h$ is continuous in $x=0$ since $h(x)\to 0$ as $x\to 0^+$. Moreover, the $h'(x) = - \frac 1{\sinh^2 x} + \frac 1{x^2}$ converges to $3$ as $x\to 0^+$, as it follows by Taylor expansion of $\sinh x$ up to the third order around $0$. The second derivative,
\[
h''(x)  = \frac {2\cosh x}{\sinh^3 x} - \frac 2{x^3}\,,
\]
verifies $h'' \le 0$, namely $\sinh^3 x > x^3 \cosh x $ for $x>0$. The latter inequality can be proven by noticing that $\sinh^3 x = \frac 14 \sinh (3x) - \frac 34 \sinh x$ and computing its Taylor expansion: its coefficients are greater than the coefficients of the Taylor expansion of $x^3 \cosh x$, but for those of $x^3$ and $x^5$ which are the same.
\end{proof}

We conclude that there exists a unique non zero solution of the equation for $\gamma$ if and only if
\[ 
3\beta J\le 1\,,
\] 
i.e., if the coupling coefficient $J$ is sufficiently small with respect to the temperature $1/\beta$. Otherwise, there exist also non constant solutions, symmetric with respect one axes, and equal but for the direction of the axis (this is the three-dimensional case of the so-called Fisher-von Mises distributions\cite{FL}).

\section{Mono-kinetic models}
\label{sec:idrocinetiche}

The flocking phenomena occur in the motion of discrete systems, composed by groups of many individuals. On the other hand, these phenomena focus on collective properties of the system, i.e., the physical quantities of interest do not depend on the details concerning the motion of single individuals. Therefore, it is possible to make easier the theoretical and mathematical study of these phenomena by introducing simplified models, which correspond to suitable large scale descriptions (continuum limits) of the underlying discrete systems. We have already encountered an example of continuum limit in Sec.~\ref{sec:vs.bagno}, where the collective behavior of the system with constant communication matrix entries  is approximated for large values of  $N$ by the non-linear Fokker-Planck equation \eqref{fpf}. 

\subsection{Mean field models}
\label{sub:mfe}

We start by analyzing the mean-field limits of the system when the interaction with the thermal bath is absent. We assume that the communication matrix entries depend only on the inter-particle distances, i.e.,
\[
n_{ij} = K(|\ve x_i - \ve x_j|)\,,
\]
where $K:\R^+\to \R^+$ is a non increasing positive regular function, with compact support. Eqs.~\eqref{eq:particelle} then reduces to the deterministic system,
\begin{equation}
\label{eq:h-particelle}
\left\{
\begin{aligned}
&\dot {\ve x}_i = \ve v_i\,,\\
&\dot {\ve v}_i = \ve s_i \wedge \ve v_i\,,\\
&\dot {\ve s}_i =\frac J{v^2} \ve v_i \wedge \ve w_i\,,\\
&\ve w_i = \frac 1N \sum_{j=1}^N  K(|\ve x_i -\ve x_j|) \ve v_j\,,\\
& |\ve v_i| = v\,, \quad \ve v_i \cdot \ve s_i = 0\,.
\end{aligned}
\right. 
\end{equation}
To characterize the behavior of the system in the mean-field limit $N\to + \infty$, we apply the theory firstly developed by Dobrushin \cite{dobrushin}, see in particular \cite{CCJ,carrillo-wassertein} where the specific case of kinetic models of collective motion is concerned. First of all, we introduce the notion of weak solutions of the expected mean-field equation (MFE).

\begin{definition}[Weak formulation of the MFE]
Given a probability measure $\mu_0(\de \ve x, \de \ve v, \de \ve s)$ on $\R^3\times TS_v^2$, a measure $\mu_t$ is called a weak solution of the MFE with initial datum $\mu_0$ if $\mu_t$ is continuous in $t$ with respect to the weak topology of measures, and, for any $\varphi\in C^{\infty}_c(\R^3\times TS_v^2)$,
\begin{equation}
\label{eq:mfe-debole0}
\int \varphi(\ve x, \ve v, \ve s) \mu_t(\de \ve x, \de \ve v,\de \ve s) = \int \varphi(\ve X_t, \ve V_t, \ve S_t) \mu_0(\de \ve x, \de \ve v, \de \ve s)\,,
\end{equation}
where
\begin{equation}
\label{eq:mfe-debole1}
(\ve X_t,\ve V_t,\ve S_t)=(\ve X_t(\ve x,\ve v,\ve s), \ve V_t(\ve x,\ve v,\ve s),\ve S_t(\ve x,\ve v,\ve s))
\end{equation}
denotes the solution to the Cauchy problem,
\begin{equation}
\label{eq:mfe-debole2}
\left\{
\begin{aligned}
&\dot {\ve X}_t = \ve V_t\,,\\
&\dot {\ve V}_t = \ve S_t \wedge \ve V_t\,,\\
&\dot {\ve S}_t = \frac J{v^2} \ve V_t \wedge \ve w_{\mu_t}(\ve X_t)\,,\\
& (\ve X_0,\ve V_0,\ve S_0) = (\ve x,\ve v,\ve s) \in \R^3\times TS_v^2\,,
\end{aligned}
\right.
\end{equation}
with
\begin{equation}
\label{eq:wmu}
\ve w_{\mu_t}(\ve y) = \int K(|\ve y - \ve x|) \ve v \, \mu_t(\de \ve x,\de \ve v, \de \ve s)\,.
\end{equation}
\end{definition}

\begin{theorem}[Existence and uniqueness of solutions]
For any  probability measure $\mu_0(\de \ve x, \de \ve v, \de \ve s)$ on $\R^3\times TS_v^2$ with bounded support, there exists a unique global weak
solution $\mu_t$ to the MFE, in the sense of Eqs.~\eqref{eq:mfe-debole0}, \eqref{eq:mfe-debole1}, \eqref{eq:mfe-debole2}, and \eqref{eq:wmu}. Moreover, $\mu_t$ has bounded support on $\R^3\times TS_v^2$ and it is weakly continuous with respect to the initial datum $\mu_0$.
\end{theorem}

\begin{proof}
We briefly outline the proof, which is standard and it is consequence of the a priori estimates $|\ve V_t|= v$, $|\ve X_t| \le |\ve X_0|+vt$, and $|\ve S_t|\le |\ve S_0|+Jv^{-2}\|K\|_{\infty}t$. Given $T>0$, for $t\in [0,T]$ the vector field in the right-hand side of Eq.~\eqref{eq:mfe-debole2} is uniformly bounded and uniformly Lipschitz if $\mu_t$ has compact support. Moreover, $\mu \mapsto \ve w_\mu$ is weakly continuous. These facts ensure the existence of the flow associated to Eqs.~\eqref{eq:mfe-debole2} for a given weakly continuous trajectory $\{\mu_t\}$ of compactly supported measures. Moreover, given $\mu_0$ with compact support, there exists a unique statistical solution $\{\tilde \mu_t\}$ to Eqs.~\eqref{eq:mfe-debole2} with initial condition $\tilde \mu_0 = \mu_0$, which is a weakly continuous family of measures, whose supports are confined in a bounded region which depends only on the support of $\mu_0$ in view of the a priori estimates. The solution to the MFE is a fixed point in the space of measure valued trajectories of the map $\{\mu_t\} \to \{\tilde \mu_t\}$ (for given $\mu_0$). This can be obtained as the limit in the Monge-Kantorovich-Rubinstein distance of the sequence of measure valued trajectories obtained by iterating the map. From this construction, one also shows uniqueness of the fixed point and its continuity with respect to the initial datum $\mu_0$.
\end{proof}

The mean-field limit is now an immediate consequence of the continuity with respect to initial data of the MFE. To each solution $\{(\ve x^N_i(t),\ve v^N_i(t),\ve s^N_i(t)\}_{i=1}^N$ to Eqs.~\eqref{eq:h-particelle} we associate the empirical measure
\[
\mu_t^N(\de \ve x, \de \ve v, \de \ve s) = \frac 1N \sum_{i=1}^N \delta_{\ve x^N_i(t)}(\de \ve x) \, \delta_{\ve v^N_i(t)}(\de \ve v)  \, \delta_{\ve s^N_i(t)}(\de \ve s)\,,
\]
where we inserted the superscript $N$ to emphasize the dependence on the size $N$ of the system.

\begin{corollary}[The mean field limit]
Let $\mu_0$ be a probability measure on $\R^3\times TS_v^2$ with compact support and let $(\ve x_i^N(0),\ve v_i^N(0),\ve s_i^N(0))_{i=1}^N$ be a family of initial condition for Eqs.~\eqref{eq:h-particelle} such that
\[
\mu_0^N  \, \underset{\rm weak}{\longrightarrow} \, \mu_0\, \quad \text{as } N\to + \infty.
\]
Then, for any $t>0$,
\[
\mu_t^N \, \underset{\rm weak}{\longrightarrow} \, \mu_t\, \quad \text{as } N\to + \infty.
\]
where $\mu_t$ solves the MFE with initial condition $\mu_0$.
\end{corollary}

The corollary follows by noticing that the empirical measure $\mu_t^N$ is a weak solution to the MFE.

\begin{remark}
\label{rem:sumf}
We can consider different normalizations for the interaction term, e.g.,
\begin{equation}
\label{eq:Knormalizzazione}
\ve w^q_i(\ve x_i) = \frac{\frac 1N \sum_j K(|\ve x_i - \ve x_j|) }{\left(\sum_j K(\ve x_i - \ve x_j)/N\right)^q}
\,,
\end{equation}
with $q \in [0,1]$. If $q = 0$, $\ve w^q_i$ growths linearly in the local density and this is the choice made in \cite{ism} for the ISM model. If $q = 1$, $\ve w^q_i$ is a weighted average of the velocities of the particles near $\ve x_i$. In the limit $N\to +\infty$ we formally obtain Eqs.~\eqref{eq:mfe-debole0}, \eqref{eq:mfe-debole1}, and \eqref{eq:mfe-debole2}, where $\ve w_{\mu_t}$ is replaced by
\begin{equation}
\label{eq:wmueta}
\ve w^q_{\mu_t}(\ve x) = \frac{ \int K(|\ve x - \ve y|) \ve v \, \mu_t(\de \ve y,\de \ve v, \de \ve s)}{\left(\int K(|\ve x - \ve y|) \mu_t(\de \ve y,\de \ve v, \de \ve s)\right)^q}\,.
\end{equation}
It is easy to show that $|\ve w_{\mu_t}^q|\le \|K\|_{\infty}^{1-q}$, but $\ve x \mapsto \ve w_{\mu_t}^q(\ve x)$ is not Lipschitz if $\mu_t$ is a generic measure. Moreover, $\ve w_{\mu_t}^q(\ve x)$ is not weakly continuous with respect to $\mu_t$. For these reasons, the convergence of the particles system to the mean-field equation is in this case not obvious.
\end{remark}

\begin{remark}
\label{rem:mfe1}
In \cite{top1,top2}, experimental data are shown which suggest that the interaction between two animals in a group is not weighted with the distance, but it is weighted with the rank. In the mean-field limit for the IS model this feature can be modelled by replacing $\ve w_\mu$ with
\begin{eqnarray}
\label{eq:rank}
&&\ve w_{\mu}^\mathrm{rank}(\ve x) = \int T(M_{\ve x, |\ve x - \ve y|}(\rho)) \ve v \, \mu (\de \ve y, \de \ve v, \de \ve s) \,, \nonumber \\
&&M_{\ve x,R}(\rho) = \int_{|\ve x-\ve z|<R} \rho(\de \ve z)\,, \nonumber \\
&&\rho(\cdot )= \int \mu(\cdot  \,,\de \ve v, \de \ve s)\,,
\end{eqnarray}
where again $T$ is a non increasing positive regular function supported in $[0,1]$, $\rho$ is the spatial density, $M_{\ve x,R}$ is the mass  within distance $R$ from $\ve x$. If $\mu$ is an empirical measure, $M_{\ve x_i,R}$ is $1/N$ the number of particles in within distance $R$, then particle $j$ contributes to $\ve w_\mu^\mathrm{rank}(\ve x_i)$ with a coefficient which depends on $M_{\ve x_i,|\ve x_i - \ve x_j|}$, proportional to the position of the particle $j$ in the ranking of the closer particles to $\ve x_i$.

It is easy to prove global existence and uniqueness of solutions for the MFE with the choice $\ve w_\mu = \ve w_\mu^\mathrm{rank}$ as in \eqref{eq:rank}, in the class of absolutely continuous measures with respect to the Lebesgue measure on $\R^3\times S_v^2\times \R^3$ (or $\R^3\times TS_v^2$) with $L^{\infty}$ density. But $\ve w_\mu^\mathrm{rank}$ is not weakly continuous in $\mu$, and if $\mu$ is an empirical measure then $\ve x \mapsto \ve w_\mu^\mathrm{rank}(\ve x)$ is also discontinuous. Therefore, it is not clear in which sense the MFE can be the limit of the particles system. A result in this direction has been proved in \cite{haskovec}, where the convergence to a mean-field kinetic equation is obtained for a smoothed version of the model. Namely, it is sufficient to replace in $M_{|\ve x|,r}(\mu)$ the measure $\mu$ with its  convolution with a regular positive compactly supported kernel, to recover the needed continuity  in the Monge-Kantorovich-Rubinstein distance and the Lipschitz regularity in the spatial variable of $\ve w_{\mu}^\mathrm{rank}$. We also mention that a kinetic Boltzmann equation for a stochastic particle model with rank based interaction has been obtaind in \cite{deg-pul} using the BBGKY hierarchy.
\end{remark}

\begin{remark}
\label{rem:mfe2}
For a regular spatial density $\rho(\ve x)$,
\[
\begin{aligned}
&\int T(M_{\ve x,|\ve x - \ve y|}(\rho))
\rho(\ve y) \de \ve y =
\int_0^{+\infty} \de R \,T(M_{\ve x,R}(\rho))
\int_{|\ve x - \ve y| =R}\rho(\ve y) \sigma(\de \ve y)
\\
&= \int_0^{+\infty} \de R \,T(M_{\ve x,R}(\rho))
\frac {\de \phantom{a}}{\de R} M_{\ve x,R}(\rho) =
\int_0^{1} \de M \, T(M)\,,
\end{aligned}
\]
so normalizing the interaction with a power of $\int T(M_{\ve x,|\ve x - \ve y|}(\rho)) \rho(\ve y) \de \ve y $ changes the equation only for a coefficient.
\end{remark}

\subsection{Mono-kinetic models in the zero-range limit}
\label{sub:hydrokin}

In general, to describe the macroscopic behaviour of mean field models, one can either analyze its hydrodynamic limits or study a special class of solutions, the so-called mono-kinetic solutions. The latter approach can be more justified in some context when dealing with individuals of biology. Indeed, in many cases, in the natural macroscopic scale there are too few individuals (animals) for unit volume to justify a hydrodynamic description.

\begin{definition}[Mono-kinetic solutions]
A mono-kinetic solution of the MFE is a triple $(\rho, \ve u, \vev \varsigma) = (\rho(\ve x,t), \ve u(\ve x,t), \vev \varsigma(\ve x,t)) \in \R \times S_v^2 \times \R^3$ such that the measure
\[
\mu_t(\de \ve x,\de \ve v,\de \ve s) = \rho(\ve x,t) \de \ve x  \, \delta_{\ve u(\ve x,t)}( \de\ve v)\, \delta_{\vev \varsigma(\ve x,t)}(\de \ve s) 
\]
is a weak solution of the MFE.
\end{definition}

\begin{proposition}
\label{prop:nk}
The triple $(\rho, \ve u, \vev \varsigma)$ is a regular mono-kinetic solution of MEF if and only if it solves the following system of PDEs,
\begin{equation}
\label{eq:hydro3-eulero}
\left\{
\begin{aligned}
&\pa_t \rho + \dive (\rho \ve u ) = 0\,, \\
&\pa_t \ve u  + (\ve u \cdot \nabla) \ve u = \vev \varsigma\wedge \ve u\,,\\
&\pa_t \vev \varsigma  + (\ve u \cdot \nabla) \vev \varsigma = \frac J{v^2} \ve u \wedge \ve w\,,\\
&\ve w(\ve x,t) = \int_{\R^3} K(|\ve x-\ve y|) \ve u(\ve y,t) \rho(\ve y,t) \de \ve y\,.
\end{aligned}\right.
\end{equation}
\end{proposition}
The proof follows by direct inspection. We remark that $|\ve u|$ and $\ve u\cdot \vev \varsigma$ are conserved along the flux generated by $\ve u$, so that they are constant if they do not depend on $\ve x$ at time $t=0$.

We can further simplify the model taking the limit when the radius of interaction vanishes, as also suggested in \cite{ism}. To this aim, we rescale by a small factor $\eps$ the argument of the communication weights, i.e., we replace $K(|\ve x - \ve y|)$ by $K(|\ve x - \ve y|/\eps)$, and then we look for the leading term in the $\eps$-expansion of $\ve u \wedge \ve w_\mu$.

\begin{lemma}
\label{lem:exp}
If $\rho(\ve x)$ and $\varphi(\ve x)$ are regular functions then
\[
\begin{aligned}
&\int_{\R^3} K(|\ve x - \ve y|/\eps) \rho(\ve y) (\varphi(\ve y)-\varphi(\ve x)) \de \ve y \\ 
& \qquad = \eps^3 \int_{\R^3} K(|\ve z|) \rho(\ve x + \eps \ve z)
(\varphi(\ve x + \eps \ve z) -\varphi(\ve x)) \de \ve z\\
& \qquad = \eps^5 b_K \left( \frac 12 \rho(\ve x) \Delta \varphi (\ve \ve x) + \nabla \rho(\ve x) \cdot \nabla \varphi(\ve x)\right) + O(\eps^7)\,,
\end{aligned}
\]
where
\[
b_K = \frac 13 \int_{\R^3}|\ve z|^2  K(|\ve z|)\de \ve z\,.
\]
\end{lemma}

The proof is achieved straightforwardly, by expanding, with respect to $\eps$ up to the third order, the term inside the integral, and then using that $\int K(|\ve z|) z_i \de \ve z = 0$, $\int K(|\ve z|) z_iz_j \de \ve z = b_K\delta_{ij}$, and $\int K(|\ve z|) z_iz_jz_k \de \ve z = 0$, for any $i,j,k=1,2,3$.

As
\[
\ve u(\ve x,t) \wedge \ve w(\ve x,t) =\ve u(\ve x,t) \wedge  \int_{\R^3} K(|\ve x-\ve y|/\eps) (\ve u(\ve y,t) - \ve u(\ve x,t)) \rho(\ve y,t) \de \ve y\,,
\]
from the previous lemma we conclude that the interaction term is of order $\eps^5$. Therefore, by rescaling $J\to J/\eps^5$ and defining $j = Jb_K/2$,
in the limit $\eps \to 0$ we obtain the mono-kinetic equations in the zero-range interaction limit,
\begin{equation}
\label{eq:hydro3-lap}
\left\{
\begin{aligned}
&\pa_t \rho + \dive (\rho \ve u ) = 0\,, \\
&\pa_t \ve u  + (\ve u \cdot \nabla)  \ve u = \vev \varsigma\wedge \ve u\,, \\
&\pa_t \vev \varsigma  + (\ve u \cdot \nabla) \vev \varsigma = \frac{j}{v^2}\ve u \wedge (\rho \Delta \ve u + 2 D \ve u \nabla \rho) = \frac{j}{v^2} \ve u \wedge \Delta (\rho \ve u)\,,\\
\end{aligned}
\right.
\end{equation}
where we have used that, for any component $u_i$ of $\ve u$,
\[
\rho \Delta u_i + 2 \nabla \rho \cdot \nabla u_i = \Delta (\rho u_i ) - (\Delta \rho) u_i\,.
\]

\begin{remark}
\label{rem:divf}
Eq.~\eqref{eq:hydro3-lap}$_3$, written in divergence form takes the form,
\[
\pa_t (\rho \vev \varsigma)   + \dive (\rho \vev \varsigma \otimes \ve u ) = \frac j{v^2} \rho \ve u \wedge \Delta (\rho \ve u)\,.
\]
Moreover, also the right-hand-side is a divergence, since, setting $\ve p = \rho \ve u$,
\[
\ve p \wedge \Delta \ve p = \dive ( \Omega(\ve p) (D \ve p)^t)\,,
\]
where $\Omega$ is defined as in Eq.~\eqref{eq:omega}. As a consequence, $\int_{\R^3} \rho \vev \varsigma$ is a conserved quantity. In the two dimensional case, this structure of conservation law for $\rho \vev \varsigma$ has been already found in \cite{ism}, and it is easier to write, starting from \eqref{eq:2d} below.
\end{remark}

\begin{remark}
\label{rem:boh}
The same expansion can be done in the case of the generalized interaction $\ve w_{\mu}^q$ defined in \eqref{eq:Knormalizzazione}. The right-hand side of the third equation in \eqref{eq:hydro3-lap} is then replaced by
\[
\frac j{v^2\rho^q} \ve u \wedge \Delta (\rho \ve u)\,.
\]
\end{remark}

Also the rank models admits a limit equation of this type. Analogously to what done before: rescaling now $\ve w_{\mu}^\mathrm{rank} $ in \eqref{eq:rank}, the interaction term becomes
\[
\ve u (\ve x)\wedge \ve w(\ve x) = u (\ve x)\wedge \int T(M_{\ve x,|\ve x - \ve y|}(\rho)/\eps^3) (\ve u(\ve x)-\ve u(\ve y)) \, \rho(\ve y) \de \ve y\,,
\]
and we search for the leading term in the $\eps$-expansion.

\begin{lemma}
Let $\rho$ and $\varphi$ be regular functions, with $\rho(\ve x) \neq 0$.  Then
\[
Q_\eps := \int T(M_{\ve \ve x, |\ve x - \ve y|}(\rho)/\eps^3) \rho(\ve y) (\varphi(\ve y) - \varphi(\ve x))\de \ve y =  \eps^5 \frac{b_T}{\rho(\ve x)^{5/3}}\Delta \varphi(\ve x)+O(\eps^7)\,,
\]
where 
\[
b_T =  \frac 13 \int |\vev \zeta|^2 T(4\pi |\vev \zeta|^3) \de \vev \zeta\,.
\]
\end{lemma}

\begin{proof}
First we observe that
\[
M_{\ve x, |\ve x - \ve y|}(\rho) = 4\pi \rho(\ve x) |\ve x - \ve y|^3 + O(|\ve x - \ve y|^5)\,,
\]
then, using $\ve z$ such that $\ve y = \ve x + \eps \ve z $ as integration variable,
\[
Q_\eps = \eps^3 \int T( 4\pi \rho(\ve x) |\ve z|^3 + O(\eps^2)) \rho(\ve x + \eps \ve z) (\varphi(\ve x + \eps \ve z) -\varphi(\ve x))\de \ve z\,.
\]
Since $T$ depends only on $|\ve z|$, the odd terms in the development vanish, so that
\[
\begin{aligned}
Q_\eps & = \frac{\eps^5}3 \left( \frac 12 \rho(\ve x) \Delta \varphi(\ve x) + \nabla \rho(\ve x)  \cdot \nabla \varphi(\ve x) \right) \int |\ve z|^2 T( 4\pi \rho(\ve x) |\ve z|^3 ) \de \ve z + O(\eps^7)  \\
& =\frac{\eps^5b_T}{\rho(\ve x)^{5/3}} \left(\frac 12 \rho(\ve x) \Delta \varphi(\ve x) + \nabla \rho(\ve x)  \cdot \nabla \varphi(\ve x) \right) + O(\eps^7)\,,
\end{aligned}
\]
where the last identity follows after introducing $\vev \zeta = \rho(\ve x)^{1/3} \ve z$ as variable of integration.
\end{proof}

To summarize, the zero-range mono-kinetic IS model is defined by the following system of PDEs,
\begin{equation}
\label{eq:hydro-tutte}
\left\{
\begin{aligned}
&\pa_t \rho + \dive (\rho \ve u ) = 0\,, \\
&\pa_t \ve u  + (\ve u \cdot \nabla)  \ve u = \vev \varsigma\wedge \ve u\,,\\
&\pa_t \vev \varsigma  + (\ve u \cdot \nabla) \vev \varsigma = \frac j{v^2\rho^{q}}\ve u \wedge \Delta (\rho \ve u)\,,\\
\end{aligned}
\right.
\end{equation}
where $q \in [0,1]$ in the case of distance based interaction, and $q = 5/3$ in the case of rank based interaction.

\begin{theorem}
\label{mono-0r}
The system \eqref{eq:hydro-tutte} is hyperbolic, in the sense that small perturbations move with relative speed $\sqrt{j\rho^{1-q}}/v$ with respect to $\ve u$, orthogonally to $\ve u$.
\end{theorem}

\begin{proof}
Let  $\bar \rho>0$, $\bar {\ve u}$ be constant fields, with $\bar \rho > 0$, $|\bar {\ve u}| = 1$, Then $(\bar \rho, \bar {\ve u}, \ve 0)$ is a stationary solution. The linearized equations around this stationary solution are
\[
\left\{
\begin{aligned}
&\pa_t \rho_1 + (\bar {\ve u} \cdot \nabla) \rho_1 =
-\bar{\rho} \dive \ve u_1\,,\\
&\pa_t \ve u_1 + (\bar {\ve u} \cdot \nabla) \ve u_1 = 
\varsigma_1\wedge \bar {\ve u}\,,\\
&\pa_t \vev \varsigma_1  + (\bar {\ve u} \cdot \nabla)  \vev \varsigma_1
= \frac j{v^2\bar {\rho}^q}
\bar {\ve u} \wedge (\bar \rho \Delta \ve u_1 + \bar {\ve u} \Delta \rho_1)
= \frac j{v^2}\bar {\rho}^{1-q} \bar {\ve u} \wedge \Delta \ve u_1\,.\\
\end{aligned}
\right.
\]
Notice that $(\ve u_1 \cdot \bar {\ve u})$ and $(\vev \varsigma_1 \cdot \bar {\ve u})$ are constants (but  nothing we can say about $\vev \varsigma_1 \cdot \vev u_1$). Only the components of $\ve u_1$ and $\vev \varsigma_1$ in the orthogonal direction to $\ve u$ can propagate,
\[
\pa_t^2 \ve u_1( \ve x+ \bar {\ve u} t, t) = \frac j{v^2} \bar {\rho}^{1-q} P^\perp_{\bar {\ve u}} \Delta \ve u_1 ( \ve x+ \bar {\ve u} t, t)\,.
\]
\end{proof}

\subsection{Mono-kinetic rotating solutions}
\label{sub:ruotanti}

In this section, we study the particular class of mono-kinetic solutions which are invariant with respect to one axis $\ve e$ and with $\vev \varsigma = \varsigma \ve e$, where
$\varsigma$ is a scalar quantity. Without loss of generality we choose $\ve e = \ve e_3$. Assuming also that $\ve u \cdot \vev \varsigma =0$, the equations of motion are simplified, since $\ve u$ can be expressed in terms of a rotation field, 
\[
\ve u(\ve x,t) = \binom{v\ve U (\vartheta(x_1,x_2,t))}{0}\,, \quad \text{where} \quad \ve U(\vartheta) = \begin{pmatrix}-\sin \vartheta  \\ \cos \vartheta \end{pmatrix}.
\]
Eqs.~\eqref{eq:hydro-tutte} become
\begin{equation}
\label{eq:2d}
\left\{
\begin{aligned}
&\pa_t \rho + v\dive (\rho \ve U) = 0\,, \\
&\pa_t \vartheta + v \ve U \cdot \nabla \vartheta = \varsigma\,, \\
&\pa_t \varsigma + v\ve U \cdot \nabla \varsigma = \frac j{\rho^{1+q}} \dive (\rho^2 \nabla \vartheta)\,,
\end{aligned}
\right.
\end{equation}
where $\dive$ and $\nabla$ are the divergence and the gradient with respecy to $(x_1,x_2)\in \R^2$. Concerning Eq.~\eqref{eq:2d}$_2$, it is consequence of the fact that
\[
\pa_t \ve U = \ve U^\perp \pa_t \vartheta\,, \quad (\ve U \cdot \nabla) \ve U  = (\ve U \cdot \nabla \vartheta)\ve U^\perp \,, \quad
\ve e_3 \wedge \ve u = v\binom{\ve U^\perp}{0}\,,
\]
where $\ve U^\perp(\theta) = \ve U'(\vartheta) = - \begin{pmatrix} \cos\vartheta \\ \sin\vartheta \end{pmatrix}$. To prove Eq.~\eqref{eq:2d}$_3$, we first notice that
\[
\ve u \wedge \Delta (\rho \ve u) = v^2 (\ve U^\perp \cdot \Delta (\rho \ve U))\, \ve e_3\,.
\]
Moreover, if $\ve p = \ve p(\ve x)$ is a regular vector field in $\R^2$,
\[
\ve p^\perp \cdot \Delta \ve p = \dive ( p_1 \nabla p_2 - p_2 \nabla p_1)\,.
\]
Then, choosing $\ve p = \rho \ve U$, the term with the gradient of $\rho$
vanishes, and, since $\nabla U_1 = - U_2 \nabla \vartheta$ and $\nabla U_2 =  U_1 \nabla \vartheta$, we have $\rho \ve U^\perp \cdot \Delta (\rho \ve U) = \dive (\rho^2 \nabla \vartheta)$.

It is now easy to find rotating stationary solutions.

\begin{theorem}
\label{teo:rot}
Let $r, \varphi$ be polar coordinates in the plane, $g(r)$ a regular function with support away from $r=0$. Then
\[
\big(\rho(r,\varphi), \vartheta(r,\varphi), \varsigma(r,\varphi) \big)= \left(g(r), \varphi, \frac 1{vr}\right)
\]
is a stationary solution of  Eqs.~\eqref{eq:2d}.
\end{theorem}

\begin{proof} 
The gradient of a function which does not depend on $\phi$, such as $\rho$ and $\varsigma$, is orthogonal to $\ve U$, while $\nabla \vartheta = \frac 1r \ve U(\varphi)= \frac 1r \ve U(\vartheta(r,\varphi))$. Then
\[
\begin{aligned}
&\ve U \cdot \nabla \rho = 0\,, \quad \ve U \cdot \nabla \vartheta = \frac 1r \,, \\
&\dive \ve U (\vartheta) = \ve U^\perp \cdot \nabla \vartheta=0\,, \quad
\dive (\rho^2 \nabla \vartheta ) = \dive( \rho^2/r \ve U(\varphi)) = 0\,.
\end{aligned}
\]
\end{proof}

\subsection{Mono-kinetic line solutions}
\label{sub:1d}

In this last section we consider the particular class of one dimensional mono-kinetic solutions in $\R^3$. It is useful to express the equations in a Lagrangian formalism.

\begin{definition}[Mono-kinetic line solutions]
\label{def:mls} 
Let $(t,z) \mapsto (\ve x_t(z),\ve v_t(z), \ve s_t(z)) \in \R^3\times \R^3 \times \R^3$ be a smooth map. The triple $(\ve x_t(z),\ve v_t(z), \ve s_t(z))$ is a mono-kinetic line solution of the MFE if, for some parameter $\lambda>0$,
\[
\mu_t(\de \ve x, \de \ve v, \de \ve s) = \lambda \int \de z \,\delta_{\ve x_t(z)}(\de \ve x) \delta_{\ve v_t(z)}(\de\ve v) \delta_{\ve s_t(z)}(\de \ve s)
\]
is a weak solutions of the MFE.
\end{definition}

Clearly, the triple $(\ve x_t(z),\ve v_t(z), \ve s_t(z))$ is a mono-kinetic line solution if and only if 
\[
\begin{split}
\ve x_t(z) & = \ve X_t(\ve x_0(z), \ve v_0(z), \ve s_0(z))\,,\qquad
\ve v_t(z)  = \ve V_t(\ve x_0(z), \ve v_0(z), \ve s_0(z))\,, \\
\ve s_t(z) & = \ve S_t(\ve x_0(z), \ve v_0(z), \ve s_0(z))\,, 
\end{split}
\]
with $(\ve X_t,\ve V_t,\ve S_t)$ solution to Eq.~\eqref{eq:mfe-debole2} with initial conditions $(\ve x_0(z), \ve v_0(z), \ve s_0(z))$. Moreover, we  consider the more general case discussed in Remark \ref{rem:sumf}, with $\ve w_{\mu_t}$ replaced by $\ve w_{\mu_t}^q$ as in \eqref{eq:wmueta}, that for mono-kinetic line solutions is given by
\begin{equation}
\label{eq:wmu1}
\ve w_{\mu_t}^q(\ve x) = \frac{\int K(|\ve x - \ve x_t(z)|) \ve v_t(z) \de z}{\left(\int K(|\ve x - \ve x_t(z)|) \de z\right)^q}\,.
\end{equation}

In what follows, we indicate with the prime the derivative with respect the one dimensional parameter $z$. In particular, the linear density of a mono-kinetic line solution is given by
$\lambda/|\ve x_t'|$. 

To derive the zero-range interaction limit of this class of solution, we proceed as before, by computing the leading term of the interaction \eqref{eq:wmu1}.

\begin{lemma}
\label{lem:lse}
Let $z\to \ve x(z)$ be a regular curve with $\ve x(0)= 0$, and $z\mapsto \varphi(z)\in \R$ be a smooth function with $\varphi(0) = 0$. Then,
\[
I_\eps := \int K(|\ve x(z)|/\eps) \varphi(z) \de z = \left. \frac{\eps^3b_2}2 \frac{\de}{\de z} \left( \frac{\varphi'(z)}{|\ve x'(z)|^3}\right) \right|_{z=0} + O(\eps^5)\,,
\]
with $b_2:=\int K(|z|) z^2 \de z$.
\end{lemma}

\begin{proof}
We write $I_\eps = \eps\int K(|\ve x(\eps z)|/\eps) \varphi(\eps z) \de z$ and insert the expansions,
\[
\begin{aligned}
&\varphi(\eps z) = \eps z \varphi'(0) + \frac 12 \eps^2 z^2 \varphi''(0) + \frac 16 \eps^3 z^3 \varphi'''(0) + O(\eps^4)\,, \\
& \ve x(\eps z)/\eps = z \ve x'(0)  + \frac 12 \eps z^2 \ve x''(0) + \frac 16 \eps^2 z^3 \ve x'''(0)  + O(\eps^3)\,, \\
& |\ve x(\eps z)|/\eps = |z| |\ve x'(0)| + \eps z |z| \frac {\ve x'(0) \cdot \ve x''(0)}{2|\ve x'(0)|} + C \eps^2 z^3  + O(\eps^2)\,,
\end{aligned}
\]
where $C=C(\ve x'(0), \ve x''(0), \ve x'''(0))$. Then, since
\[
\begin{split}
& \int K(|z||\ve x'|) z \de z = \int K(|z||\ve x'|) z^3 \de z = 0\,, \\
& \int K(|z||\ve x'|) z^2 \de z = \frac{b_2}{|\ve x'|^3}\,, \quad \int K'(|z||\ve x'|) |z|^3 \de z = - \frac{3b_2}{|\ve x'|^4}\,,
\end{split}
\]
we finally have
\[
I_\eps = \frac {\eps^3b_2}2 \left(
\varphi'' - 3 \frac {\ve x'\cdot \ve x''}{|\ve x'|^2} \varphi'\right) +
O(\eps^5 )= \frac {\eps^3b_2}2 \left(\frac {\varphi'}{|\ve x'|^3}\right)'+O(\eps^5)\,,
\]
where the derivatives are computed at $z=0$.
\end{proof}

Similarly, letting $b_0 :=  \int K(|z|) \de z $, we have
\[
\int K(|\ve x(z)|/\eps) \de z = \frac{\eps b_0}{|\ve x'(0)|} + O(\eps^3)\,.
\]
Therefore, by rescaling $J\to J/\eps^{2-q}$ and taking the limit $\eps \to 0$ in the term $\ve w^q(\ve x_t(z))$, we finally obtain the equations of motion in the zero-range limit, 
\begin{equation}
\label{eq:linea}
\left\{
\begin{aligned}
&\dot {\ve x}_t(z) = \ve v_t(z)\,,\\
&\dot {\ve v}_t(z) = \ve s_t(z)\wedge\ve v_t(z)\,, \\
&\dot {\ve s}_t(z) = \frac{j \lambda^{1-q}|\ve x_t'(z)|^q}{v^2} \ve v_t(z) \wedge \frac{\de}{\de z} \frac{\ve v_t'(z)}{|\ve x_t'(z)|^3}\,,
\end{aligned}
\right.
\end{equation}
with $j = J b_2/(2b_0)$.

We can also consider the case of rank based interaction, obtaining the same equation with $q=3$, as it follows from the next lemma.

\begin{lemma}
\label{lem:wrl}
Let $z\to \ve x(z)$ be a regular curve with $\ve x(0)= 0$, and $z\mapsto \varphi(z)\in \R$ be a smooth function with $\varphi(0) = 0$. Let also
\[
M_\gamma = \lambda \int  \fcr\{|\ve x(z)|<\gamma\}\de z \,.
\]
Then,
\[
I_\eps := \int T(M_{|\ve x(z)|}/\eps) \varphi(z) \de z = \left. \frac{\eps^3b}{16\lambda^3} |\ve x'(0)|^3 \frac{\de}{\de z} \left( \frac{\varphi'(z)}{|\ve x'(z)|^3}\right) \right|_{z=0} + O(\eps^4)\,,
\]
with $b:=\int T(z) z^2  \de z$.
\end{lemma}

\begin{proof}
For $|\ve x(z)|$ small,
\[
M_{|\ve x(z)|} = 2 \lambda \frac {|\ve x(z)|}{|\ve x'(0)|} + O(|\ve x(z)|^3)\,.
\]
Therefore, after the change $z\to \eps z$ of the variable of integration, we have
\[
I_\eps = \eps\int T\left( \frac {2\lambda} {|\ve x'(0)|} \frac {|\ve x(\eps z)|}{\eps}+O(\eps^2)\right) \varphi(\eps z) \de z\,.
\]
From now on, we can proceed as in the previous lemma, with $T(2\lambda |\ve x|/|\ve x'(0)|)$ in place of $K(|\ve x|)$, and the thesis follows by noticing that
\[
\int T(2\lambda z/|\ve x'(0)|) z^2 \de z = \left( \frac {|\ve x'(0)|}{2\lambda}\right)^3 \int T(|z|) z^2 \de z.
\]
\end{proof}

The mono-kinetic line equations in the zero-range limit \eqref{eq:linea}  admit any given regular curve as a solution, in the sense specified by the following theorem.

\begin{theorem}
\label{teo:lineestaz}
Let $\ve \Gamma(\alpha)$ be a regular curve in $\R^3$, parametrized with the arc-length $\alpha$. Then the triple
\[
\begin{aligned}
&\ve x_t(z) = \ve \Gamma( \gamma z + vt)\,,\\
&\ve v_t(z) = v\ve \Gamma'(\gamma z + vt)\,,\\
&\ve s_t(z) =  v\ve \Gamma' (\gamma z + vt) \wedge \ve \Gamma''(\gamma z + vt)\,,\\
\end{aligned}
\]
is a solution to Eqs.~\eqref{eq:linea} if $v^2/j = (\lambda/\gamma)^{1-q}$.
\end{theorem}

\begin{proof} Eqs.~\eqref{eq:linea}$_{1,2}$ follow immediately from the definition of the  triple. Moreover,
\[
\dot {\ve s}_t(z) = v^2 \ve \Gamma' (\gamma z + vt)  \wedge \ve \Gamma''' (\gamma z + vt) = \ve v_t(z) \wedge v \ve \Gamma''' (\gamma z + vt)\,,
\]
$|\ve x_t'(z)| = \gamma$, and
\[
\frac{\de}{\de z} \frac{\ve v_t'(z)}{|\ve x_t'(z)|^3} =  \frac{\ve v_t''(z)}{\gamma^3} = \frac{v}{\gamma} \ve \Gamma''' (\gamma z + vt)\,.
\]
Then Eq.~\eqref{eq:linea}$_3$ is satisfied provided that $v^2/j = (\lambda/\gamma)^{1-q}$.
\end{proof}

\begin{remark}
\label{rem:pu}
A solution of this kind can be seen as the motion of a line of animals that follow the path traced be the first one. Note that $\lambda/\gamma$ is the linear density. In the models distance based with $q <1$, the allowed velocity increases with the density, while in the rank based model this is inversely proportional to the density. Perhaps, this different behaviour can be useful in the applications.
\end{remark}

\begin{remark}
\label{rem:u}
It is possible to consider the zero-range limit equations for mono-kinetic line solutions of the mean field equation associated to second order systems, like Vicsek type models. But in this case the stationary fluxes of Theorem \ref{teo:lineestaz} reduce to rectilinear motions.
\end{remark}

\end{document}